\documentclass[onecolumn,english,final, twocolumn]{IEEEtran}
\usepackage[T1]{fontenc}
\usepackage[latin9]{inputenc}
\usepackage{bm}
\usepackage{amsmath}
\usepackage{amsthm}
\usepackage{amssymb}
\usepackage{graphicx}

\makeatletter
\theoremstyle{plain}
\newtheorem{thm}{\protect\theoremname}
\theoremstyle{remark}
\newtheorem{rem}{\protect\remarkname}
\theoremstyle{definition}
\newtheorem{defn}{\protect\definitionname}

\usepackage{algorithmic}

\usepackage{psfrag}
\usepackage{cite}
\usepackage{subfigure}\usepackage{algorithm}

\author{
Hao~Sun,~Xianghao~Yu,~\IEEEmembership{Senior Member,~IEEE}, and Junting~Chen,~\IEEEmembership{Member,~IEEE}
\thanks{Hao Sun and Xianghao Yu are with the Department of Electrical Engineering, City University of Hong Kong, Hong Kong (e-mail:
hao.sun@cityu.edu.hk; alex.yu@cityu.edu.hk).}\thanks{Junting Chen is with the School of Science and Engineering (SSE), Shenzhen Future Network of Intelligence Institute (FNii-Shenzhen), and
Guangdong Provincial Key Laboratory of Future Networks of Intelligence, The Chinese University of Hong Kong, Shenzhen, Guangdong 518172, China
(e-mail: juntingc@cuhk.edu.cn).}}

\usepackage[acronym]{glossaries}
\newcommand{\newac}{\newacronym}
\newcommand{\ac}{\gls}
\newcommand{\Ac}{\Gls}

\makeglossaries

\newac{speb}{SPEB}{square position error bound}
\newac[plural=EFIMs,firstplural=Fisher information matrices (EFIMs)]{efim}{EFIM}{Fisher information matrix}
\newac{ne}{NE}{Nash equilibrium}
\newac{mse}{MSE}{mean squared error}
\newac{toa}{TOA}{time of arrival}
\newac{tdoa}{TDOA}{time difference of arrival}
\newac{snr}{SNR}{signal-to-noise ratio}
\newac{lan}{LAN}{local area network}
\newac{psd}{PSD}{positive semidefinite}
\newac{pd}{PD}{positive definite}
\newac{wrt}{w.r.t.}{with respect to}
\newac{lhs}{L.H.S.}{left hand side}
\newac{wp1}{w.p.1}{with probability 1}
\newac{kkt}{KKT}{Karush-Kuhn-Tucker}
\newac{wlog}{w.l.o.g.}{without loss of generality}
\newac{mle}{MLE}{maximum likelihood estimation}
\newac{rssi}{RSSI}{received signal strength indicator}
\newac{mimo}{MIMO}{multiple-input multiple-output}
\newac{csi}{CSI}{channel state information}
\newac{fdd}{FDD}{frequency division duplexing}
\newac{ms}{MS}{mobile station}
\newac{bs}{BS}{base station}
\newac{bss}{BSs}{base stations}
\newac{d2d}{D2D}{device-to-device}
\newac{slnr}{SLNR}{signal-to-interference-leakage-and-noise-ratio}
\newac{ula}{ULA}{uniform linear antenna array}
\newac{pas}{PAS}{power angular spectrum}
\newac{mmse}{MMSE}{minimum mean square error}
\newac{zf}{ZF}{zero-forcing}
\newac{rzf}{RZF}{regularized zero-forcing}
\newac{as}{AS}{angular spread}
\newac{ue}{UE}{user equipment}
\newac{ues}{UEs}{user equipments}
\newac{iid}{i.i.d.}{independent and identically distributed} 
\newac{sinr}{SINR}{signal-to-interference-and-noise ratio}
\newac{tdd}{TDD}{time-division duplex}
\newac{rvq}{RVQ}{random vector quantization}
\newac{rhs}{R.H.S.}{right hand side}
\newac{mrc}{MRC}{maximum ratio combining}
\newac{cdf}{CDF}{cumulative distribution function}
\newac{a.s.}{a.s.}{almost surely}
\newac{los}{LOS}{line-of-sight}
\newac{jsdm}{JSDM}{joint spatial division and multiplexing}
\newac{map}{MAP}{maximum a posteriori}
\newac{klt}{KLT}{Karhunen-Lo\`eve Transform}
\newac{lbe}{LBE}{link bargaining equilibrium}
\newac{se}{SE}{Stackelberg equilibrium}
\newac{uav}{UAV}{unmanned aerial vehicle}
\newac{uavs}{UAVs}{unmanned aerial vehicles}
\newac{nlos}{NLOS}{non-line-of-sight}
\newac{pdf}{PDF}{probability density function}
\newac{em}{EM}{expectation-maximization}
\newac{knn}{KNN}{$k$-nearest neighbors}
\newac{svd}{SVD}{singular value decomposition}
\newac{nmf}{NMF}{non-negative matrix factorization}
\newac{umf}{UMF}{unimodality-constrained matrix factorization}
\newac{rmse}{RMSE}{rooted mean squared error}
\newac{olos}{OLOS}{obstructed line-of-sight}
\newac{mmw}{mmW}{millimeter wave}
\newac{ber}{BER}{bit error rate}
\newac{rss}{RSS}{received signal strength}
\newac{lp}{LP}{linear program}
\newac{ufw}{U-FW}{unimodal Frank-Wolfe}
\newac{utf}{UTF}{unimodality-constrained tensor factorization}
\newac{fw}{FW}{Frank-Wolfe}
\newac{iot}{IoT}{Internet-of-Things}
\newac{mae}{MAE}{mean absolute error}
\newac{crb}{CRB}{Cram\'er-Rao bound}
\newac{aoa}{AoA}{angle of arrival}
\newac{wcl}{WCL}{weighted centroid localization}
\newac{slf}{SLF}{spatial loss function}
\newac{btd}{BTD}{block-term tensor decomposition}
\newac{tps}{TPS}{thin plate spline}
\newac{nmse}{NMSE}{normalized mean squared error}
\newac{svt}{SVT}{singular value thresholding}
\newac{vae}{VAE}{variational autoencoder}
\newac{gan}{GAN}{generative adversarial network}
\newac{aod}{AOD}{angle of departure}
\newac{rbf}{RBF}{radial basis function}
\newac{loocv}{LOOCV}{leave one out cross validation}
\newac{lpr}{LPR}{local polynomial regression}
\newac{xl-mimo}{XL-MIMO}{extremely large multiple-input multiple-output}
\newac{6g}{6G}{sixth-generation}
\newac{cpd}{CPD}{canonical polyadic decomposition}
\newac{nnm}{NNM}{nuclear norm minimization}
\newac{mimo-ofdm}{MIMO-OFDM}{multiple-input multiple-output orthogonal frequency division multiplexing}
\newac{mimo-cdma}{MIMO-CDMA}{multiple-input multiple-output code division multiple access}
\newac{pad}{PAD}{power-angular-delay}
\newac{hmm}{HMM}{hidden Markov model}
\newac{dft}{DFT}{discrete time Fourier transform}
\newac{pnp}{PnP}{Plug-and-Play}
\newac{ode}{ODE}{ordinary differential equation}
\newac{gps}{GPs}{Gaussian processes}
\newac{ai}{AI}{artificial intelligence}
\newac{2d}{2D}{two-dimensional}
\newac{ddpm}{DDPM}{denoising diffusion probabilistic modeling}

\setkeys{Gin}{width=1.0\columnwidth}

\providecommand{\definitionname}{Definition}

\providecommand{\remarkname}{Remark}
\providecommand{\theoremname}{Theorem}

\makeatother

\usepackage{babel}
\providecommand{\definitionname}{Definition}
\providecommand{\remarkname}{Remark}
\providecommand{\theoremname}{Theorem}

\begin{document}
\title{Structure-Aware Near-Field Radio Map Recovery via RBF-Assisted Matrix
Completion}
\maketitle
\begin{abstract}
This paper proposes a novel structure-aware matrix completion framework
assisted by radial basis function (RBF) interpolation for near-field
radio map construction in extremely large multiple-input multiple-output
(XL-MIMO) systems. Unlike the far-field scenario, near-field wavefronts
exhibit strong dependencies on both angle and distance due to spherical
wave propagation, leading to complicated variations in received signal
strength (RSS). To effectively capture the intricate spatial variations
structure inherent in near-field environments, a regularized RBF interpolation
method is developed to enhance radio map reconstruction accuracy.
Leveraging theoretical insights from interpolation error analysis
of RBF, an inverse $\mu$-law-inspired non-uniform sampling strategy
is introduced to allocate measurements adaptively, emphasizing regions
with rapid RSS variations near the transmitter. To further exploit
the global low-rank structure in the near-field radio map, we integrate
RBF interpolation with nuclear norm minimization (NNM)-based matrix
completion. A robust Huberized leave-one-out cross-validation (LOOCV)
scheme is then proposed for adaptive selection of the tolerance parameter,
facilitating optimal fusion between RBF interpolation and matrix completion.
The integration of local variation structure modeling via RBF interpolation
and global low-rank structure exploitation via matrix completion yields
a structure-aware framework that substantially improves the accuracy
of near-field radio map reconstruction. Extensive simulations demonstrate
that the proposed approach achieves over 10\% improvement in normalized
mean squared error (NMSE) compared to standard interpolation and matrix
completion methods under varying sampling densities and shadowing
conditions.
\end{abstract}


\section{Introduction}

\Ac{6g} wireless communication is envisioned to provide ultra-high
data rates, massive device connectivity, and high precision localization
and sensing \cite{CheMit:J19,ZenCheXuWu:J24,YuXZhaHaeLet:J17}. A
key enabler of these capabilities is the deployment of \ac{xl-mimo}
systems, which increases the number of antennas by an order of magnitude
and expands the array aperture to the meter scale. This substantial
aperture expansion fundamentally alters the electromagnetic propagation
regime by enlarging the radiative near-field region, where wavefront
curvature becomes non-negligible. The boundary of this region, characterized
by the Rayleigh distance, scales quadratically with aperture size,
bringing near-field effects into practical communication ranges. As
a result, XL-MIMO systems depart from the far-field assumptions of
conventional MIMO and instead operate in a regime dominated by spherical
wavefronts.

Within the expanded near-field regime, the wireless channel exhibits
enhanced spatial resolution and finer angular discrimination due to
its spherical wave propagation. These unique physical characteristics
provide the theoretical foundation for a range of frontier applications
critical to \ac{6g} systems, such as high-precision user localization
and tracking \cite{LeiZhaWanAiB:J25}, and efficient adaptive beam
management \cite{CuiWuZluYWei:J23}. To effectively exploit these
advantages in practice, it is highly beneficial to construct an accurate
near-field radio map, which provides a powerful tool for modeling
and leveraging the complexity of the near-field environment.

Specifically, the near-field radio map enables these advanced applications
by characterizing the intricate spatial distribution of \ac{rss}
across both angle and distance. For user localization and tracking,
the near-field radio map allows observed \ac{rss} patterns to be
aligned with the map's known distributions, achieving sub-meter positioning
precision \cite{LeiZhaWanAiB:J25}. Moreover, its detailed angle-distance
profiles directly support adaptive beam management strategies, providing
the critical data needed to optimize beam steering and reduce the
overhead of channel estimation and initial access in highly directional
\ac{xl-mimo} systems \cite{CuiWuZluYWei:J23}.

Recent advances in radio map construction have spanned several categories,
including interpolation-based methods, compressive sensing methods,
tensor methods and learning-based methods. Interpolation-based methods
include Kriging \cite{SatKoyFuj:J17,SatSutIna:J21} and kernel-based
methods \cite{TegRomRam:J19,BazGia:J13,XuZhaZha:C21,HamBer:C17,XuHuaJia:J21}.
In \cite{SatSutIna:J21}, a space-frequency interpolation scheme was
developed by introducing a linear shadowing interpolation across frequencies
before applying spatial Kriging. A recent work \cite{XuHuaJia:J21}
learned a set of kernel density functions to model the radio map using
nonparametric estimation. Compressive sensing-based methods include
sparse Bayesian learning \cite{WanZhuLin:J24,WanZhuLin:J24b,GaoZhuLinMat:J25},
dictionary learning \cite{KimGia:C13}, and matrix completion \cite{MigMarDon:J11,SunChe:C21,XiaZhaHua:J23}.
The works in \cite{WanZhuLin:J24,WanZhuLin:J24b} leveraged sparse
Bayesian learning to reduce the number of required measurements while
maintaining high reconstruction accuracy. In \cite{GaoZhuLinMat:J25},
gradient descent was combined with greedy matching to enhance sampling
efficiency, followed by spatial-temporal semi-variogram modeling to
facilitate accurate reconstruction. Tensor methods \cite{ChoMicLovKro:J21,SchCavSta:C19,MalZhaXuy:C18,SheWanDin:J22,SunChe:J24,CheWanHua:J25,SunCheLuo:C24}
have also been explored for high-dimensional radio map construction.
Recent work \cite{SheWanDin:J22} proposed an orthogonal matching
pursuit based on a tensor structure to recover 3D spectrum map. The
work in \cite{CheWanHua:J25} leveraged \ac{cpd} to reconstruct
spatial-spectral-temporal map. In parallel, learning-based methods
have emerged as powerful tools for radio map construction \cite{XuLCheChePuw:C25,TimShrFux:J24,CheChe:J24,ShrFuHong:J22,LuoLiZPenChe:J25}.
In \cite{TimShrFux:J24}, a generative model combined with a tensor
structure was proposed to extract underlying structures in the map.
The authors in \cite{CheChe:J24} developed a virtual obstacle model
to characterize geometry-induced signal variation through neural networks.
In \cite{LuoLiZPenChe:J25}, a novel generative framework leveraging
conditional denoising diffusion probabilistic models was introduced
to synthesize high-quality radio map from limited data. Beyond these,
hybrid methods that combine interpolation with structure\nobreakdash-informed
modeling have gained increasing attention \cite{SunChe:J22,ZhaWan:J22,CheWanZha:J23,SunChe:J24b,DonPuWZhoFuX:C25,SunChe:C22}.
In \cite{SunChe:J22}, an uncertainty-aware matrix completion approach
was proposed through integrating \ac{lpr} to improve radio map
construction. The work in \cite{ZhaWan:J22} utilized $k$-nearest
neighbors Gaussian process regression to exploit spatial signal similarity.
In \cite{CheWanZha:J23}, interpolation was incorporated with block-term
tensor decomposition to reconstruct 3D map. A convex optimization
framework that integrates interpolation with matrix completion was
proposed in \cite{DonPuWZhoFuX:C25}, which embeds interpolation into
a low-rank matrix model and solves it using the alternating direction
method of multipliers (ADMM).

However, most of these methods are not well-suited for near-field
radio map construction due to the distinctive propagation characteristics
in the near-field region. Unlike its far-field counterpart, near-field
wavefronts exhibit strong dependence on both angle and distance, resulting
in rapid spatial variations of signal strength and distinct radial
symmetry around the antenna array. Nevertheless, since the \ac{rss}
is determined by a small number of physical parameters, primarily
range, azimuth, and the array layout, the resulting near-field radio
map exhibits an approximately low-rank structure. Unfortunately, conventional
reconstruction approaches typically fail to simultaneously exploit
both local spatial correlations and the inherent global low-rank structure
of near-field radio map, while learning-based methods require large
amounts of labeled data and suffer from limited interpretability.

Addressing these challenges is hindered by two main issues. First,
purely global low-rank models often neglect sharp, localized RSS variations,
particularly near the transmitter, which in turn makes conventional
uniform sampling highly inefficient by placing redundant samples in
smooth regions while undersampling areas of rapid change. Second,
existing hybrid methods that combine interpolation and matrix completion
lack a systematic approach to determining the optimal balance between
local interpolation accuracy and global completion consistency.

To address these limitations, we introduce a novel framework that
jointly exploits the local variation and low-rank structures for near-field
radio map construction. Our method uses regularized \ac{rbf} interpolation
as a powerful prior to assist the matrix completion process. This
RBF-assisted approach enables a theoretically-grounded, non-uniform
sampling strategy that concentrates measurements where they are most
needed. Furthermore, it incorporates a data-driven mechanism to adaptively
balance the influence of the RBF model against the low-rank assumption,
creating a robust and accurate reconstruction.

The main contributions are summarized as follows:
\begin{itemize}
\item We develop a unified approach that amalgamates regularized RBF interpolation,
which captures local near-field variations, with matrix completion
that recovers the global low-rank structure of the near-field radio
map. The regularized RBF model includes a constant term that improves
reconstruction accuracy, especially around field boundaries.
\item Motivated by theoretical RBF interpolation error analysis, which reveals
the significance of measurement density and spatial signal variation,
we propose an inverse $\mu$-law-based non-uniform sampling strategy
that places denser measurements in regions with sharp RSS variation
(e.g., near the transmitter), thereby improving spatial resolution
where it is most critical.
\item We propose an adaptive parameter estimation strategy using robust
Huberized LOOCV. By robustly estimating interpolation error from data-driven
cross-validation, this method adaptively selects the tolerance parameter
$\delta$, ensuring an optimal balance between fidelity to the local
RBF prior and enforcement of global low-rank structure.
\item Our experiments find that the proposed \ac{rbf}-driven matrix completion
framework consistently outperforms traditional interpolation and matrix
completion techniques. Regularized RBF interpolation achieves over
40\% \ac{nmse} reduction at low sampling densities, while inverse
$\mu$-law sampling improves accuracy by over 10\% compared to uniform
sampling. The adaptive tolerance parameter selection via Huberized
\ac{loocv} further enhances reconstruction stability and performance.
Overall, the method delivers more than 10\% \ac{nmse} gain under
various sampling and shadowing conditions.
\end{itemize}

The rest of the paper is organized as follows. Section \ref{sec:System-Model}
introduces the near-field channel model, describing the considered
XL-MIMO system and the associated near-field signal propagation characteristics.
Section \ref{sec:Structure-Aware-Recovery-Framewo} details our proposed
solution in a step-by-step manner. It first presents the regularized
RBF interpolation method for modeling local signal variations, followed
by the inverse $\mu$-law sampling strategy designed to optimize its
accuracy. It then integrates this RBF-based prior into a global matrix
completion model, introducing the Huberized LOOCV scheme as a robust
mechanism for fusing the two components. The computational complexity
of the proposed method is then analyzed. Section \ref{sec:Simulation-Result}
presents extensive simulation results, demonstrating the effectiveness
and robustness of the proposed methods. Finally, conclusions are provided
in Section \ref{sec:Conclusion}.

\emph{Notation:} Scalars and sets are denoted by italic letters (e.g.,
$x$ and $\mathcal{X}$). Bold italic lowercase letters (e.g., $\bm{x}$)
denote vectors, and bold italic uppercase letters (e.g., $\bm{X}$)
represent matrices. The entry at the $i$-th row and $j$-th column
of matrix $\bm{X}$ is denoted by $X_{ij}$. The symbol $\|\bm{X}\|_{*}$
refers to the nuclear norm, representing the sum of the singular values
of matrix $\bm{X}$. The transpose and conjugate transpose are denoted
by $(\cdot)^{\text{T}}$ and $(\cdot)^{\text{H}}$, respectively.
$\mathbb{R}$ and $\mathbb{C}$ denote the sets of real and complex
numbers, respectively. $\mathcal{N}(\mu,\sigma^{2})$ and $\mathcal{CN}(\mu,\sigma^{2})$
denote the real- and complex-valued Gaussian distributions with mean
$\mu$ and variance $\sigma^{2}$, respectively. $\jmath$ denotes
the imaginary unit with $\jmath^{2}=-1$. $O(x)$ means $|O(x)|/x\leq C$,
for all $x>x_{0}$ with $C$ and $x_{0}$ are positive real numbers.
The symbol $|\cdot|$ denotes the absolute value of a scalar argument
(e.g., $|x|$) and the cardinality of a set argument (e.g., $|\mathcal{X}|$).

\section{System Model\label{sec:System-Model}}

Consider a downlink narrowband near-field \ac{xl-mimo} system as
illustrated in Fig.~\ref{fig:XL-MIMO-near-filed-communication}.
A transmitter equipped with $N$ antennas, arranged in an extremely
large \ac{ula}, serves a designated area. Assume there are $M$
single antenna receivers. \Ac{wlog}, the transmitter is positioned
along the y-axis, with the $n$-th antenna located at $\left(0,\delta_{n}\lambda/2\right)$,
where $\delta_{n}=(2n-N-1)/2$ for $n=1,2,\ldots,N$, $\lambda=c/f$
denotes the wavelength, $f$ is the operating frequency, and $c$
is the speed of light.
\begin{figure}
\centering\includegraphics[width=6cm,totalheight=8cm,height=5cm]{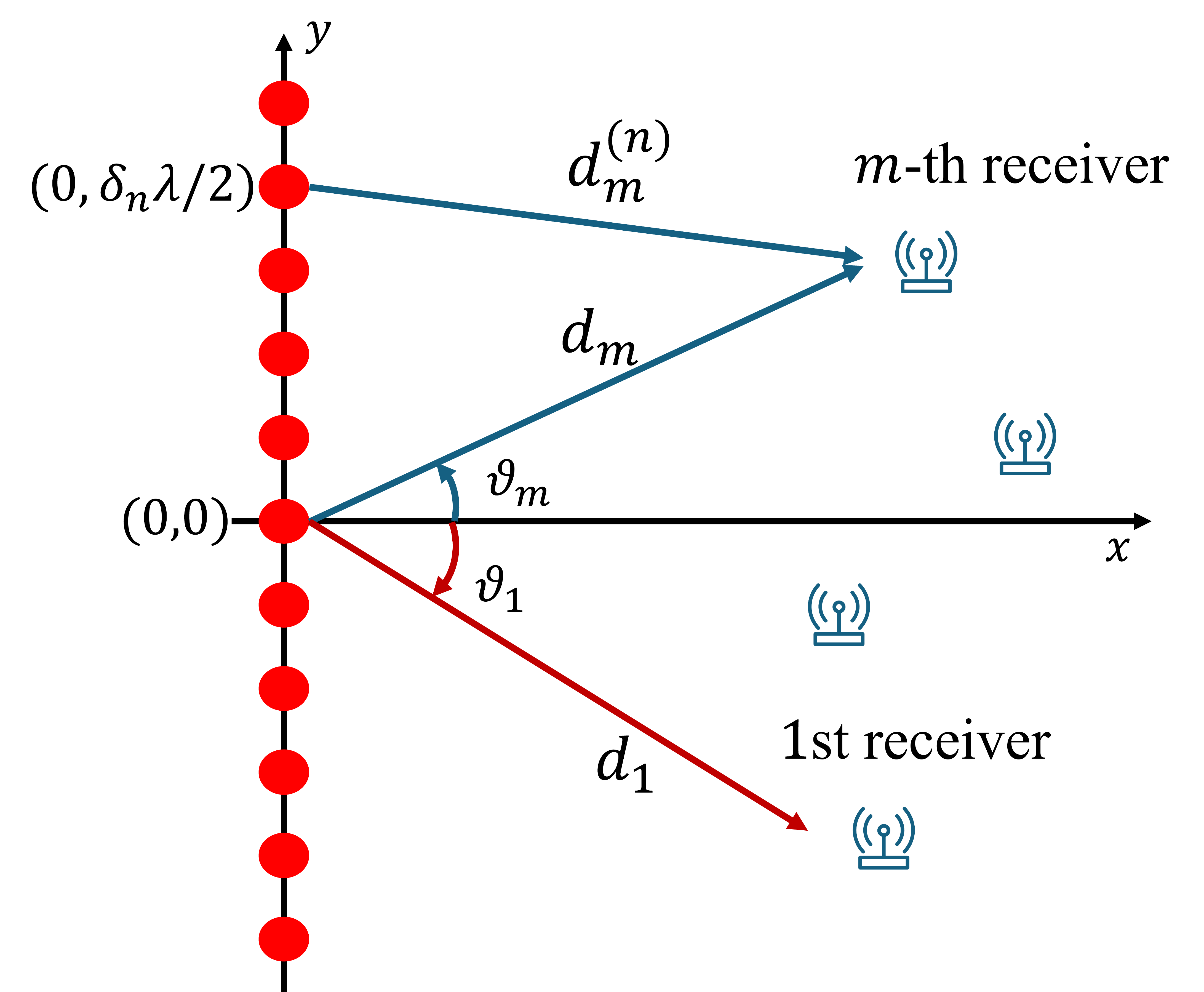}\caption{\label{fig:XL-MIMO-near-filed-communication}XL-MIMO near-field sensing
scenario.}
\end{figure}

The received signal $y_{m}$ at the $m$-th receiver is given by
\begin{equation}
y_{m}=\sqrt{P}\bm{h}_{m}^{\text{H}}\bm{v}s+n,\quad m=1,2,\cdots,M,\label{eq:}
\end{equation}
where $P>0$ is the transmit power, $\bm{h}_{m}\in\mathbb{C}^{N\times1}$
is the downlink channel vector, and $\bm{v}\in\mathbb{C}^{N\times1}$
denotes the transmit beamforming vector at the transmitter. The term
$s$ is the transmitted symbol, satisfying $\mathbb{E}[|s|^{2}]=1$
w.l.o.g., and $n\sim\mathcal{CN}\left(0,\xi^{2}\right)$ denotes the
additive receiver noise, where $\xi^{2}$ is the noise power.

The downlink channel steering vector $\bm{h}_{m}\in\mathbb{C}^{N\times1}$
is modeled as 
\begin{equation}
\bm{h}_{m}=\sqrt{\frac{N}{L}}\sum_{l=1}^{L}\beta_{l,m}\bm{a}\left(\vartheta_{l,m},d_{l,m}\right),\label{eq:channel model}
\end{equation}
where $L$ denotes the number of paths, which include the \ac{los}
and \ac{nlos} paths. The parameter $\beta_{l,m}$ represents the
path attenuation between the center of the antenna array and the $m$-th
reciever along the $l$-th path, and $\bm{a}(\vartheta_{l,m},d_{l,m})$
is the near-field array steering vector, where $\vartheta_{l,m}$
denotes the spatial angle from the center of the antenna array to
the $m$-th receiver along the $l$-th path, and $d_{l,m}$ is the
distance between the center of the antenna array and the $m$-th receiver
along the $l$-th path.

Due to the large aperture of the XL-MIMO array and the resulting near-field
propagation properties, the power of the NLOS paths is typically more
than 20 dB weaker than that of the LOS path \cite{YuXZhaHaeLet:J17,AkdLiuSamSun:J14}.
Consequently, the \ac{los} path typically dominates. Hence, \ac{wlog},
we assume $L=1$, simplifying the channel model (\ref{eq:channel model})
to 
\begin{equation}
\bm{h}_{m}=\sqrt{N}\beta_{m}\bm{a}(\vartheta_{m},d_{m}).\label{eq:LOS hm}
\end{equation}

The path attenuation $\beta_{m}$ of the \ac{los} path, including
the path loss and the shadowing effect, is modeled as 
\begin{equation}
\beta_{m}=\frac{\lambda}{4\pi d_{m}}10^{\varepsilon/20},\label{eq:beta m}
\end{equation}
where $\varepsilon$ is a random variable accounting for the effect
of shadowing in logarithmic (dB) scale with $\varepsilon\sim\mathcal{N}\left(0,\sigma^{2}\right)$.

Based on the spherical wavefront propagation model, the near-field
array steering vector $\bm{a}(\vartheta_{m},d_{m})$ is given by 
\begin{equation}
\bm{a}\left(\vartheta_{m},d_{m}\right)=\frac{1}{\sqrt{N}}\left[e^{-\jmath2\pi d_{m}^{(0)}/\lambda},\ldots,e^{-\jmath2\pi d_{m}^{(N-1)}/\lambda}\right]^{\text{T}},\label{eq:-1}
\end{equation}
where $d_{m}^{(n)}$ denotes the distance between the $m$-th receiver
and the $n$-th antenna at the transmitter, computed as
\begin{equation}
d_{m}^{(n)}=\sqrt{d_{m}^{2}+\delta_{n}^{2}\lambda^{2}/4-d_{m}\text{cos}(\vartheta_{m})\delta_{n}\lambda}.\label{eq:r m n}
\end{equation}

In the dB scale, from (\ref{eq:})\textendash (\ref{eq:beta m}),
the \ac{rss} $\gamma_{m}$ is then defined by
\begin{align}
\gamma_{m} & =10\text{log}_{10}\left(\left|\sqrt{P}\bm{h}_{m}^{\text{H}}\bm{v}s\right|^{2}\right)\nonumber \\
 & =10\text{log}_{10}\left(\left|\sqrt{\frac{P}{N}}\frac{\lambda}{4\pi d_{m}}\sum_{n=0}^{N-1}v_{n}e^{-\jmath\frac{2\pi}{\lambda}d_{m}^{(n)}}s\right|^{2}\right)+\varepsilon.\label{eq:gamma_ij}
\end{align}

To enable near\nobreakdash-field radio\nobreakdash-map construction,
we discretize the region of interest into an angular-distance grid,
yielding an $I\times J$ RSS matrix $\bm{\Gamma}\in\mathbb{R}^{I\times J}$,
where $I$ and $J$ are the numbers of quantized angular and radial
points, respectively. Each cell $(i,j)$ uniquely corresponds to a
spatial angle $\theta_{i}$ and radial distance $r_{j}$. Let $(i_{m},j_{m})$
denote the grid cell in which the $m$-th receiver lies, with spatial
angle $\theta_{i_{m}}$ and radial distance $r_{j_{m}}$. In other
words, the continuous angular $\vartheta_{m}$ and radial distance
$d_{m}$ of the $m$-th receiver are quantized to their discrete counterparts
associated with cell $(i_{m},j_{m})$. The measured RSS $\gamma_{m}$
at the $m$-th receiver is therefore assigned to cell $(i_{m},j_{m})$,
i.e., $\Gamma_{i_{m}j_{m}}=\gamma_{m}$.

Our goal in this paper is to reconstruct the complete RSS matrix $\bm{\Gamma}$
using only sparse observations $\{(\theta_{i_{m,}},r_{j_{m}},\Gamma_{i_{m}j_{m}})\}_{m=1}^{M}$
collected by a limited number of receivers.

\section{Structure-Aware Recovery Framework\label{sec:Structure-Aware-Recovery-Framewo}}

In this section, we introduce a structure-aware recovery framework
for near-field radio map reconstruction that jointly exploits the
local variation and global low-rank structures. Specifically, the
proposed approach leverages a physics-informed prior derived via RBF
interpolation, which is seamlessly integrated into a global low-rank
matrix completion formulation. We first describe the overall recovery
framework. Then, we detail three critical components: (i) the construction
of a high-fidelity local prior informed by near-field propagation
physics, (ii) a theoretically grounded non-uniform sampling strategy
that enhances construction accuracy, and (iii) an adaptive fusion
mechanism that balances the local prior and global model via selection
of the tolerance parameter.

\subsection{Overall Formulation}

Recall that $\bm{\Gamma}\in\mathbb{R}^{I\times J}$ denotes the ground
truth RSS matrix. There are sparse observations $\{\Gamma_{ij}\}$
available at a subset of entries, where the subscript $m$ below $i$
and $j$ is omitted for notation simplicity. As illustrated in Fig.~\ref{fig:(a)-Ground-truth rm},
the ground truth near-field RSS matrix exhibits a dominant low-rank
structure with energy concentrated in a few singular values, thereby
motivating the use of low-rank recovery.
\begin{figure}
\subfigure{\includegraphics[width=0.5\columnwidth]{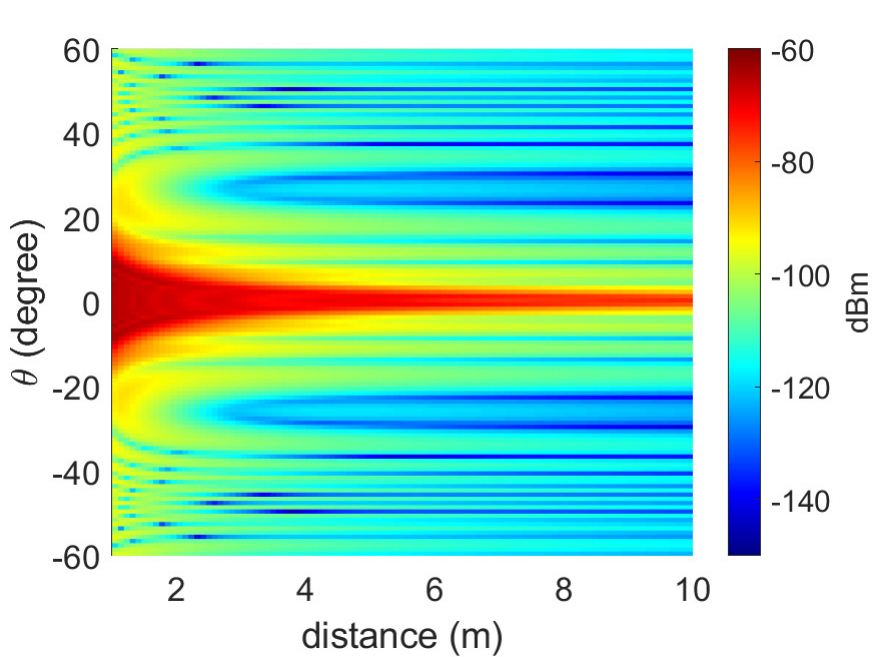}}\subfigure{\includegraphics[width=0.5\columnwidth]{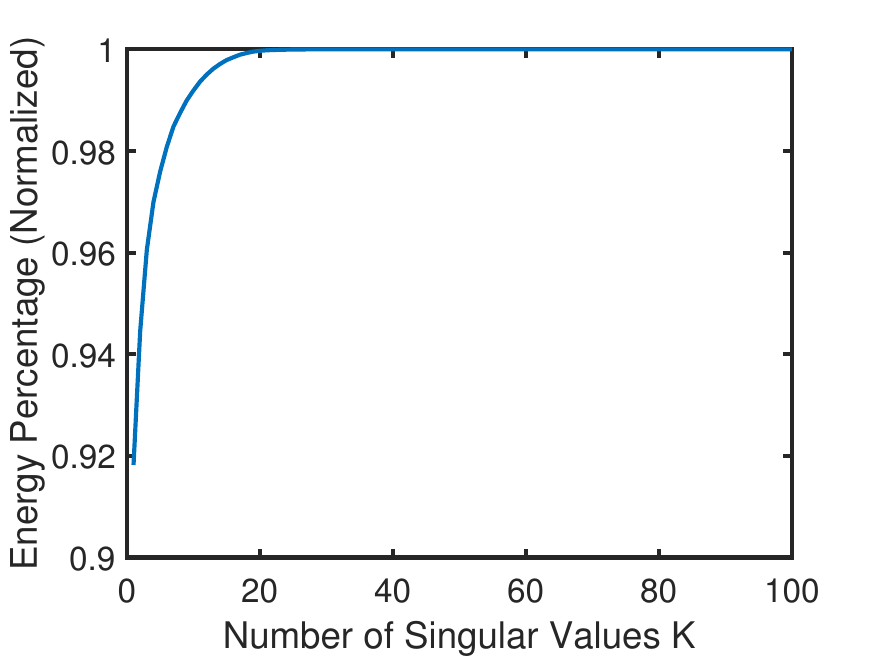}}

\caption{\label{fig:(a)-Ground-truth rm}(a) Ground truth of near-field radio
map. (b) Percentage of the sum of the first $K$ dominant singular
values over the sum of all singular values for a $100$ by $100$
near-field radio map matrix.}
\end{figure}
 In addition to its low-rank nature, the RSS matrix also exhibits
strong local correlations arising from the spatial continuity of radio
signal propagation. To exploit both properties, we formulate a regularized
\ac{nnm} problem as
\begin{equation}
\begin{aligned}\underset{\bm{Z}\in\mathbb{R}^{I\times J}}{\text{minimize}}\quad & \|\bm{Z}\|_{*}\\
\text{subject to}\quad & \hat{\bm{\Gamma}}=f(\mathcal{P}_{\Omega}(\bm{\Gamma}))\\
 & |Z_{ij}-\hat{\Gamma}_{ij}|\leq\delta,\ \forall(i,j)
\end{aligned}
,\label{eq:MC2}
\end{equation}
where $\bm{Z}\in\mathbb{R}^{I\times J}$ is the optimization variable
representing the reconstructed RSS matrix, and $\|\bm{Z}\|_{*}$ promotes
a low-rank structure. $f(\mathcal{P}_{\Omega}(\bm{\Gamma}))$ denotes
an interpolation operator applied to the observations indexed by $\Omega$
with cardinality $|\Omega|=M$. The sampling operator $\mathcal{P}_{\Omega}:\mathbb{R}^{I\times J}\to\mathbb{R}^{I\times J}$
is defined by:
\[
[\mathcal{P}_{\Omega}(\bm{\Gamma})]_{ij}=\begin{cases}
\Gamma_{ij}, & (i,j)\in\Omega\\
\text{unknow}, & \text{otherwise}
\end{cases}.
\]
The parameter $\delta>0$ controls the fidelity of $\bm{Z}$ to the
interpolated prior $\hat{\bm{\Gamma}}$. A small $\delta$ forces
the recovered matrix to closely follow the interpolated prior, emphasizing
local signal variations, whereas a large $\delta$ relaxes the constraint,
allowing $\bm{Z}$ to deviate from the prior and rely more heavily
on \ac{nnm}, thereby enhancing the enforcement of global low-rank
structure. Therefore, an appropriate choice of $\delta$ is crucial
for balancing local accuracy and global consistency.

The effectiveness of the proposed framework hinges on 1) the accurate
construction of the prior $\hat{\bm{\Gamma}}$ to reflect the spatial
characteristics of near-field RSS; 2) the strategic selection of sampling
locations to maximize measurement informativeness; 3) and the data-driven
yet theoretically grounded tuning of the fusion parameter $\delta$,
which governs the integration level between the interpolated prior
and the recovered low-rank matrix. In the following, we shall present
these three key design aspects.

\subsection{Local Prior Construction via Regularized RBF Interpolation\label{subsec:Modeling-Near-Field-Physics}}

In this subsection, we construct a high-fidelity interpolated prior
$\hat{\bm{\Gamma}}$ that captures the spatial characteristics of
near-field RSS.

Direct interpolation over \ac{2d} spatial domains typically requires
dense sampling, since under spherical wave propagation, the \ac{rss}
exhibits sharp variations along the angular direction, as illustrated
in Fig. \ref{fig:The-RSS-versus direction with fixed r}. To alleviate
this issue, we conduct a near-field sensitivity analysis to justify
adopting a 1D interpolation scheme along the radial direction, rather
than performing full 2D interpolation that demands dense angular sampling.
\begin{figure}
\includegraphics{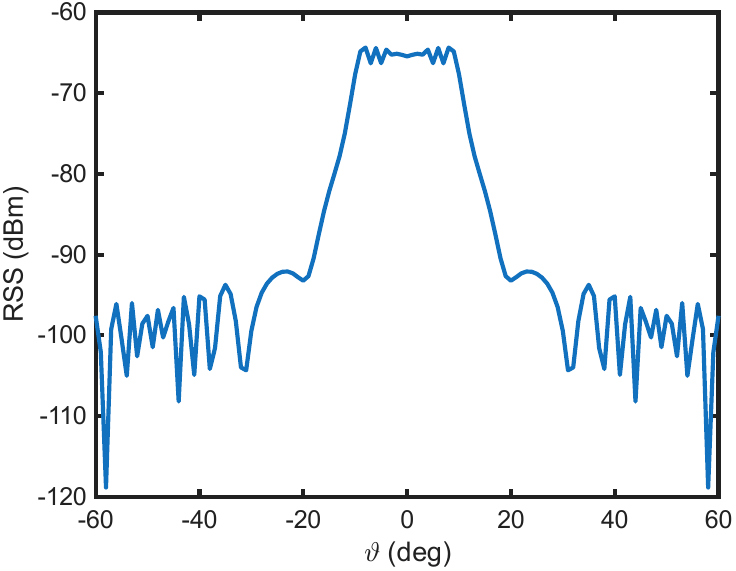}

\caption{\label{fig:The-RSS-versus direction with fixed r}RSS versus direction
$\vartheta$ at a fixed distance $d$. The RSS exhibits a sharp main
lobe with multiple sidelobes.}
\end{figure}

We define the near-field array factor in (\ref{eq:gamma_ij}) as $S(d,\vartheta)=\frac{\lambda}{4\pi d}\sum_{n=0}^{N-1}v_{n}e^{-\jmath\tfrac{2\pi}{\lambda}d^{(n)}(d,\vartheta)}$.
\begin{thm}[Angular vs. Radial Sensitivity]
\label{thm:The-partial-derivatives}The partial derivatives of the
near-field array factor magnitude $|S(d,\vartheta)|$ \ac{wrt} angular
and radial perturbations satisfy
\begin{align}
\frac{\partial|S|}{\partial\vartheta} & \leq\text{cos}(\phi_{\vartheta})\frac{5\lambda|\sin\vartheta|}{16d_{\text{min}}}N^{2}\label{eq: sensitivity angle}\\
\frac{\partial|S|}{\partial d} & \leq\text{cos}(\phi_{d})\left(\frac{\lambda}{4\pi d^{2}}+\frac{1}{2d}\right)N,\label{eq:sensitivity distance}
\end{align}
where $\cos\phi_{x}$ is a phase-alignment factor and is independent
of the array size $N$, $d_{min}$ represents the minimal distance
from the measurement location to the array.
\end{thm}
\begin{proof} Please see Appendix \ref{sec:Proof-of-Lemma non smooth}.
\end{proof}

Theorem \ref{thm:The-partial-derivatives} characterizes the sensitivity
of the near-field array factor magnitude $|S(d,\vartheta)|$ \ac{wrt}
angular and radial perturbations, highlighting critical distinctions
in their relative behavior.
\begin{rem}
Angular Sensitivity: The angular derivative exhibits an $O(N^{2})$
envelope, scaled by $\lambda|\sin\vartheta|/(4d_{\min})$ and modulated
by a phase-alignment factor $\cos\phi_{\vartheta}\in[-1,1]$. Consequently,
even small angular perturbations can induce large RSS fluctuations,
consistent with the narrow main lobes and dense sidelobes in Fig.~\ref{fig:The-RSS-versus direction with fixed r}.
Compared with the radial derivative (of order $O(N)$), the angular
sensitivity grows faster with $N$, making angular interpolation the
sampling bottleneck as $N$ increases.
\end{rem}
\begin{rem}
Radial Smoothness: In contrast, the radial derivative grows only linearly
with $N$ and decays smoothly with $1/d$, indicating much gentler
RSS variations along the radial dimension. This is visually corroborated
by Fig.~\ref{fig:RSS-versus-distance}, where the RSS changes gradually
with distance. Hence, to cope with the severe angular sensitivity
while maintaining manageable sampling overhead, we strategically perform
one-dimensional interpolation along the radial direction for each
fixed angle. This approach leverages the smoother radial structure
to construct high-fidelity near-field radio maps without incurring
excessive angular sampling burden. 
\begin{figure}
\includegraphics{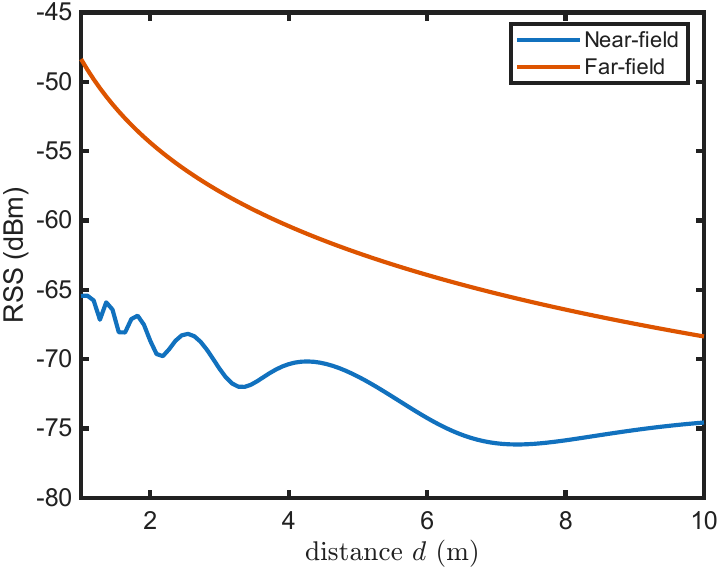}\caption{\label{fig:RSS-versus-distance}RSS versus distance $d$ along a fixed
spatial angle $\vartheta=0^{\circ}$ under $N=256$ antenna elements
in the near- and far-field models. The near-field RSS exhibits more
rapid fluctuations due to spherical wavefront effects, in contrast
to the smoother decay observed in the approximation by far-field model.}
\end{figure}
\end{rem}
Therefore, rather than performing full two-dimensional interpolation
over the angular-radial domain, we focus on reconstructing RSS values
along the radial dimension for each fixed spatial angle $\vartheta$.

Let $\mathcal{S}_{i}=\{(\theta_{i},r_{j},\Gamma_{ij})\}$ denote the
set of available measurements at a fixed spatial angle $\theta_{i}$,
where $j\in\Omega_{i}$. Here, $\Omega_{i}=\{j\mid(\theta_{i},r_{j},\Gamma_{ij})\in\mathcal{S}_{i}\}$
represents the set of indices corresponding to radial measurement
locations associated with the fixed spatial angle $\theta_{i}$.

To capture the spatial variation of RSS along the radial dimension,
an appropriate interpolation scheme is required. Classical methods
such as \ac{lpr}, and \ac{knn} rely on the choice of a local window
size or the number of neighbors $k$. In \ac{lpr}, the model order
is also hard to choose: a low order yields lower accuracy, while a
high order may overfit, hindering their effective use. \ac{rbf} interpolation,
by contrast, provides a flexible kernel-based framework; different
kernels can be matched to different scenarios, and regularization
offers robustness to noise. Therefore, to model the near-filed \ac{rss}
variation, we adopt a regularized \ac{rbf} interpolation approach.

For each fixed angle $\theta_{i}$, the RSS function $\rho_{i}(d)$
is modeled as
\begin{equation}
\rho_{i}(d)=\sum_{j=1}^{|\Omega_{i}|}\lambda_{j}\,\phi\left(|d-r_{j}|\right)+c,\label{eq:rho_i}
\end{equation}
where $\{\lambda_{j}\}$ are interpolation weights, $\phi(\cdot)$
is the RBF kernel (we use the multiquadric kernel, to be detailed
in Section \ref{sec:Simulation-Result}), and $c$ is a constant term.
The inclusion of $c$ mitigates boundary errors and stabilizes the
solution near domain edges.

To solve for the coefficients $\{\lambda_{j}\}$ and $c$ in \eqref{eq:rho_i},
we enforce the interpolation conditions
\begin{equation}
\rho_{i}(r_{j})=\Gamma_{ij},\quad j=1,2,\ldots,|\Omega_{i}|,\label{eq:interpolation_conditions}
\end{equation}
together with the constraint
\begin{equation}
\sum_{j=1}^{|\Omega_{i}|}\lambda_{j}=0\label{eq:zero_sum_constraint}
\end{equation}
 to ensure uniqueness of the solution.

For computational efficiency, we reformulate these conditions using
matrix notation. Let $\ensuremath{|\Omega_{i}|=K}$ and define the
following variables:
\begin{itemize}
\item The RBF kernel matrix $\bm{\Phi}\in\mathbb{R}^{K\times K}$, with
entries
\begin{equation}
\Phi_{mn}=\phi\left(|r_{m}-r_{n}|\right).
\end{equation}
\item The measurement vector
\begin{equation}
\bm{\gamma}_{i}=[\Gamma_{i1},\Gamma_{i2},\dots,\Gamma_{iK}]^{\mathrm{T}}\in\mathbb{R}^{K}.
\end{equation}
\item The RBF coefficient vector
\begin{equation}
\bm{\lambda}=[\lambda_{1},\lambda_{2},\dots,\lambda_{K}]^{\mathrm{T}}\in\mathbb{R}^{K},
\end{equation}
\item The all-one vector $\mathbf{1}_{K}=[1\ 1\ \cdots\ 1]^{\mathrm{T}}\in\mathbb{R}^{K}.$
\end{itemize}
Then, the interpolation conditions and the zero-sum constraint in
\eqref{eq:interpolation_conditions} and \eqref{eq:zero_sum_constraint}
can be expressed compactly as the linear system:
\begin{equation}
\begin{bmatrix}\bm{\Phi} & \mathbf{1}_{K}\\
\mathbf{1}_{K}^{\text{T}} & 0
\end{bmatrix}\begin{bmatrix}\bm{\lambda}\\
c
\end{bmatrix}=\begin{bmatrix}\bm{\gamma}_{i}\\
0
\end{bmatrix}.\label{eq:rbf_linear_system}
\end{equation}

Given that the multiquadric kernel is strictly positive definite,
$\bm{\Phi}$ is invertible. The solution can be expressed in closed
form as
\begin{align}
c & =\left(\mathbf{1}_{K}^{\text{T}}\bm{\Phi}^{-1}\mathbf{1}_{K}\right)^{-1}\left(\mathbf{1}_{K}^{\text{T}}\bm{\Phi}^{-1}\bm{\gamma}_{i}\right),\label{eq:c_exp}\\
\bm{\lambda} & =\bm{\Phi}^{-1}\left(\bm{\gamma}_{i}-\mathbf{1}_{K}\left(\mathbf{1}_{K}^{\text{T}}\bm{\Phi}^{-1}\mathbf{1}_{K}\right)^{-1}\mathbf{1}_{K}^{\text{T}}\bm{\Phi}^{-1}\bm{\gamma}_{i}\right).\label{eq:lambda}
\end{align}

\subsection{Adaptive Sampling for RBF Interpolation}

In this subsection, we allocate sampling locations to enhance interpolation
accuracy under a fixed measurement budget. Due to spherical wave propagation,
the near-field RSS exhibits steep radial gradients near the transmitter
and smoother variations at larger distances. Combined with the RBF
error bound, this motivates allocating more samples to regions with
high variability. To this end, we develop a non-uniform sampling strategy
based on the inverse $\mu$-law transformation, which concentrates
sampling density near the transmitter.

\subsubsection{Interpolation Error Bound Analysis}

To draw the theoretical grounded insight on the choice of sampling
strategy, we analyze the interpolation error associated with multiquadric
RBF. The following definitions and theorem establish a rigorous foundation
for sampling strategy design.
\begin{defn}[Fill Distance]
Let $R=\{r_{1},\cdots,r_{N}\}$ be a set of points in a bounded domain
$\mathcal{D}$. The fill distance $h$ of $R$ \ac{wrt} $\mathcal{D}$
is defined as
\[
h=\underset{d\in\mathcal{D}}{\text{sup}}\underset{1\leq j\leq N}{\text{min}}|d-r_{j}|,
\]
which quantifies how effectively the points in $R$ cover the domain
$\mathcal{D}$. Intuitively, the fill distance $h$ represents the
radius of the largest empty ball that can fit among the measurement
points.
\end{defn}
\begin{defn}[Essentially Bounded]
Let $f:\mathcal{D}\to\mathbb{R}$ be a measurable function. The essential
supremum norm of $f$ is defined as
\[
\|f\|_{L_{\infty}(\mathcal{D})}=\inf\left\{ C_{0}\geq0:|f(d)|\leq C_{0}\text{ a.e. on }\mathcal{D}\right\} ,
\]
where \textquotedblleft a.e.\textquotedblright{} denotes almost everywhere.
This norm measures the smallest bound that holds almost everywhere
in $\mathcal{D}$.
\end{defn}
Let $\mathcal{N}_{\phi}(\mathcal{D})$ denote the native space of
the reproducing kernel $\phi$. The following theorem, adapted from
\cite{Wen:B04}, provides an upper bound for the interpolation error
using RBF interpolation.
\begin{thm}[Error bound]
\label{thm:Suppose-that-} Consider $\phi(d)=(1+d^{2})^{\beta}$,
and define $f(d)=\phi(\sqrt{d})$. Then for every integer $\ell\ge\lceil\beta\rceil$,
the $\ell$th derivative of $f$ satisfies
\begin{equation}
\left|f^{(\ell)}(d)\right|\le\ell!C_{1}^{\ell},\quad\forall d\ge0,\label{eq:derivative condition}
\end{equation}
where $C_{1}=(1+d)^{\beta/\ell-1}(1+|\beta+1|)$ is a constant.

In addition, there exists a constant $c>0$ such that for any sufficiently
small fill distance $h$, the interpolation error between a function
$f\in\mathcal{N}_{\phi}(\mathcal{D})$ and its interpolant $\rho$
is bounded by
\[
\left\Vert f-\rho\right\Vert _{L_{\infty}(\mathcal{D})}\leq e^{-c/h}|f|_{\mathcal{N}_{\phi}(\mathcal{D})},
\]
where $|\cdot|_{\mathcal{N}_{\phi}(\mathcal{D})}$ is the native seminorm,
closely related to the magnitude of higher-order derivatives of $f$.
\end{thm}
\begin{proof} Please see Appendix \ref{sec:Proof-of-Proposition}.
\end{proof}

Since the multiquadric RBF chosen in \eqref{eq:rho_i} is precisely
of the form $\phi(d)=(1+\epsilon d^{2})^{\beta}$ with $\beta=1/2$
and $\epsilon=1$, Theorem~\ref{thm:Suppose-that-} becomes directly
applicable, thereby providing explicit interpolation error bounds
for the multiquadric RBF interpolation.

We highlight two critical insights from this error analysis in the
following.
\begin{rem}
The interpolation error bound derived in Theorem \ref{thm:Suppose-that-}
indicates that as RSS measurements become denser, i.e., $h\to0$,
interpolation error decreases exponentially. Therefore, RBF interpolation
can achieve high accuracy if the measurements are sufficiently dense.
Consequently, the fill distance $h$ directly governs the achievable
interpolation accuracy.
\end{rem}
\begin{rem}
If the near-field RSS function exhibits pronounced spatial variability,
the seminorm $|f|_{\mathcal{N}_{\phi}(\mathcal{D})}$ increases, leading
to higher interpolation errors. One effective strategy to mitigate
this is to reduce the fill distance $h$ by increasing the measurement
density. However, in practical settings where the number of measurements
is constrained, a more feasible solution is to adopt a non-uniform
sampling strategy. This strategy allocates more samples to regions
with high signal variation and fewer samples to smoother regions,
thereby optimizing accuracy under a limited measurements.
\end{rem}

\subsubsection{Inverse $\mu$-Law Non-Uniform Sampling}

Motivated by the preceding error analysis and the observation that
near-field RSS exhibits steep gradients near the transmitter and smoother
variations at greater distances, we propose a non-uniform sampling
strategy based on the inverse $\mu$-law transformation. This approach
effectively concentrates samples near the transmitter, where the interpolation
error is most sensitive to spatial resolution.

Given a scalar input $x\in[0,1]$ and a companding parameter $\mu>0$,
the $\mu$-law transformation $F(x)$ is defined as
\begin{equation}
F(x)=\frac{\ln(1+\mu x)}{\ln(1+\mu)}.
\end{equation}

Directly applying the forward $\mu$-law transformation compresses
the lower portion of the domain and stretches the upper portion. However,
this contradicts the spatial requirements in near-field scenarios,
where higher resolution is needed near the transmitter (i.e., at the
lower end of the spatial domain). To address this issue, we instead
employ the inverse $\mu$-law transformation:
\[
F^{-1}(y)=\frac{(1+\mu)^{y}-1}{\mu},\quad y\in[0,1],
\]
which maps uniformly spaced values $y$ into a set of samples that
are concentrated near zero.

The non-uniform sampling procedure is implemented as follows:
\begin{itemize}
\item Generate $N$ independent uniform samples $u_{i}\sim\mathcal{U}(0,1)$,
uniformly distributed over $[0,1]$.
\item Apply the inverse $\mu$-law transformation to obtain
\begin{equation}
y_{i}=\frac{(1+\mu)^{u_{i}}-1}{\mu}.
\end{equation}
\item Map the transformed samples $y_{i}\in[0,1]$ to the target spatial
interval $[z_{0},z_{1}]$ via linear scaling:
\begin{equation}
r_{i}=z_{0}+y_{i}(z_{1}-z_{0}).
\end{equation}
\end{itemize}
The resulting spatial sample set $\{r_{i}\}$ is thus non-uniformly
distributed over $[z_{0},z_{1}]$, with a concentration of samples
near $z_{0}$, i.e., in close proximity to the transmitter.

The clustering effect is controlled by the parameter $\mu$, where
larger values yield stronger concentration near the lower end of the
domain. This non-uniform allocation efficiently reduces interpolation
error in regions with high RSS variability while maintaining overall
sampling efficiency.

\subsection{Adaptive Fusion of RBF Prior and Low-Rank Model\label{subsec:Fusing-the-RBF}}

In this subsection, we effectively fuse the RBF-interpolated prior
with the global low-rank structure discussed in \eqref{eq:MC2} to
ensure robust reconstruction. In practice, excessive reliance on the
RBF prior may propagate local interpolation errors, while an overly
loose constraint may disregard informative local variations. The tolerance
parameter $\delta$ governs this trade-off by controlling the allowable
deviation from the prior. We embed the RBF prior as inequality constraints
in an \ac{nnm} problem and adaptively determine $\delta$ through
a Huberized \ac{loocv} procedure.

\subsubsection{Integration of RBF Prior into Low-Rank Recovery}

Classical matrix completion methods leverage the assumption that the
underlying data matrix is low-rank, aiming to recover it from a small
subset of observed entries.

After constructing the prior $\hat{\bm{\Gamma}}$, we integrate it
into the low-rank recovery framework. This involves solving the main
optimization problem \eqref{eq:MC2} and adaptively setting the tolerance
parameter $\delta$. Based on the RBF interpolation result, \eqref{eq:MC2}
can be reformulated as:
\begin{equation}
\begin{aligned} & \underset{\bm{Z}}{\text{minimize}} &  & \|\bm{Z}\|_{*}\\
 & \text{subject to} &  & |Z_{ij}-\hat{\Gamma}_{ij}|\leq\delta,\ \forall(i,j)\\
 &  &  & \hat{\Gamma}_{ij}=\sum_{k\in\Omega_{i}}\lambda_{i,k}\phi_{j,k}+c_{i}\\
 &  &  & \bm{\lambda}_{i}=\bm{\Phi}_{i}^{-1}\left(\bm{\gamma}_{i}-\mathbf{1}_{K}\left(\mathbf{1}_{K}^{\text{T}}\bm{\Phi}_{i}^{-1}\mathbf{1}_{K}\right)^{-1}\mathbf{1}_{K}^{\text{T}}\bm{\Phi}_{i}^{-1}\bm{\gamma}_{i}\right)\\
 &  &  & c_{i}=\left(\mathbf{1}_{K}^{\text{T}}\bm{\Phi}_{i}^{-1}\mathbf{1}_{K}\right)^{-1}\left(\mathbf{1}_{K}^{\text{T}}\bm{\Phi}_{i}^{-1}\bm{\gamma}_{i}\right).
\end{aligned}
\label{eq:MC}
\end{equation}
Here, $\phi_{j,k}$ denotes the radial basis function evaluated at
location $r_{j}$ with respect to center $r_{k}$, and $\delta$ controls
the allowable deviation between the recovered matrix $\bm{Z}$ and
the RBF prior. The coefficients $c_{i}$ and $\bm{\lambda}_{i}$ are
from \eqref{eq:c_exp}\textendash \eqref{eq:lambda}.

The resulting problem (\ref{eq:MC}) is a convex optimization program
comprising a nuclear norm objective and linear inequality and equality
constraints. It can be efficiently solved using the standard interior
point method \cite{PotWri:J00}.

\subsubsection{Adaptive Parameter Selection via Huberized LOOCV}

The tolerance parameter $\delta$ plays a pivotal role in balancing
the influence of the RBF prior and the global low-rank constraint.
A small $\delta$ may lead to overfitting to the prior, whereas a
large $\delta$ may cause the recovery to ignore valuable local structure.

To automatically determine a proper $\delta$, a Huberized \ac{loocv}
approach is proposed. In standard LOOCV, each observation is removed
from the training set in turn; the model is then refitted using the
remaining samples, and the prediction at the held-out location yields
an out-of-sample error. Specifically, we apply LOOCV to the RBF interpolation
and obtain a sequence of prediction errors $\{e_{k}\}_{k=1}^{K}$,
where each error $e_{k}$ is defined as
\begin{equation}
e_{k}=\Gamma_{ik}-\rho_{i,-k}(r_{k}),\label{eq:-7-1-1}
\end{equation}
where $\rho_{i,-k}(r_{k})$ denotes the interpolant constructed by
excluding the $k$-th measurement from the dataset $\mathcal{S}_{i}=\{(\theta_{i},r_{j},\Gamma_{ij})\}$.
The interpolant is defined as
\[
\rho_{i,-k}(r_{k})=\sum_{\substack{j=1,j\neq k}
}^{K}\lambda_{j}\,\phi\left(\|r_{k}-r_{j}\|\right)+c,
\]
where the coefficients $\lambda_{j}$ and constant $c$ are computed
using the same interpolation procedure described in (\ref{eq:rbf_linear_system})\textendash (\ref{eq:lambda}),
but based only on the reduced measurement set $\Omega_{i}\backslash\{k\}$.

To calibrate the tolerance $\delta$, we estimate a robust typical
error level from the \ac{loocv} residuals using the Huber estimator.
This choice yields a stable summary for $\delta$ while preventing
a few large residuals from dominating.

Let $\{e_{k}\}_{k=1}^{K}$ denote the \ac{loocv} errors. The Huber
estimator $\hat{\mu}$ is defined as the minimizer of the aggregated
Huber loss:
\begin{equation}
\hat{\mu}=\arg\min_{\mu\in\mathbb{R}}\sum_{k=1}^{K}g\bigl(e_{k}-\mu\bigr),\label{eq:huber}
\end{equation}
where $g(\cdot)$ is the Huber loss function given by
\begin{equation}
g(r)=\begin{cases}
\frac{1}{2}r^{2}, & |r|\leq\varsigma_{H},\\[6pt]
\varsigma_{H}\bigl(\lvert r\rvert-\tfrac{1}{2}\varsigma_{H}\bigr), & |r|>\varsigma_{H},
\end{cases}
\end{equation}
and $\varsigma_{H}>0$ is a threshold that controls the transition
from a quadratic penalty to a linear penalty. Equation (\ref{eq:huber})
can be solved via an iterative re-weighted least squares procedure.
In iteration $t$, we update
\begin{equation}
\mu^{(t+1)}=\frac{\sum_{k=1}^{K}w_{k}^{(t)}e_{k}}{\sum_{k=1}^{K}w_{k}^{(t)}},
\end{equation}
where
\begin{equation}
w_{k}^{(t)}=\begin{cases}
1, & \text{if }\bigl|e_{k}-\mu^{(t)}\bigr|\leq\varsigma_{H},\\[6pt]
\dfrac{\varsigma_{H}}{\bigl|e_{k}-\mu^{(t)}\bigr|}, & \text{otherwise}.
\end{cases}
\end{equation}

Upon convergence, we set the tolerance based on this robust center,
$\delta=\hat{\mu}$. This data-driven choice ties $\delta$ to the
observed out-of-sample interpolation error of the RBF prior, avoiding
both overfitting and underfitting.

\subsection{Complexity Analysis\label{sec:Complexity-Analysis}}

The proposed RBF-driven matrix completion method consists of two main
components: (\emph{i}) a regularized RBF interpolation stage and (\emph{ii})
an NNM-based matrix completion stage. Each component has distinct
computational characteristics.

Let $I$ and $J$ denote the angular and radial resolutions of the
reconstructed radio map, respectively, i.e., the target RSS matrix
is $\bm{Z}\in\mathbb{R}^{I\times J}$. Interpolation is performed
independently along each angular slice $\theta_{i}$, for $i=1,2,\cdots,I$
with $\Omega_{i}$ represents the set of available measurement indices
at angle $\theta_{i}$, and assume $|\Omega_{i}|=K\ll J$. For each
$\theta_{i}$, constructing the RBF kernel matrix $\bm{\Phi}\in\mathbb{R}^{K\times K}$
requires $O(K^{2})$ operations. Solving the linear system in (\ref{eq:rbf_linear_system})
to obtain $(\bm{\lambda},c)$ requires $O(K^{3})$ operations. Thus,
across all $I$ angular directions, the total interpolation cost is
$O(IK^{3})$.

To refine the interpolated map, we solve the NNM-based matrix completion
problem (\ref{eq:MC}) , where the interpolated values serve as soft
constraints. For an $I\times J$ matrix, the completion cost is: $O((I+J)\log^{3}(I+J))$
\cite{JaiNet:B15}. The tolerance parameter $\delta$ is selected
via the Huberized LOOCV method. This requires re-solving the interpolation
$K$ times per angular slice, resulting in an overall LOOCV complexity
of $O(IK^{4}).$

Hence, the total computational complexity of the proposed method is
$O(IK^{4}+(I+J)\log^{3}(I+J)).$

The overall complexity is essentially the additive combination of
the baseline RBF interpolation and matrix completion modules, each
offering complementary benefits in terms of spatial modeling capabilities.

\section{Simulation Results\label{sec:Simulation-Result}}

We evaluate the performance of the proposed RBF-assisted matrix completion
method under a near-field millimeter-wave \ac{xl-mimo} scenario.
Unless stated otherwise, all simulations are performed under the following
settings.

\subsection{Simulation Setup}

The transmitter is equipped with a \ac{ula} of $N=256$ antennas,
with an inter-element spacing of $\lambda/2=0.0015\,\text{m}$. This
results in a total aperture of $D=N\lambda/2=0.384\,\text{m}$. The
carrier frequency is $f=100\,\text{GHz}$ (wavelength $\lambda=0.003\,\text{m}$),
yielding a Rayleigh distance of $Z=2D^{2}/\lambda\approx98.3\,\text{m}$.
The \ac{ula} is aligned with the y-axis and centered at the origin.
The transmitter transmits a unit-power symbol $s=1$, and the total
transmit power is normalized to $P=1$. We consider an omnidirectional
beamformer $\bm{v}=\frac{1}{\sqrt{N}}\bm{1}_{N}$ to focus on the
spatial distribution of RSS.

The area of interest spans angles $\theta\in[-80^{\circ},80^{\circ}]$
and radial distances $r\in(0,10]\,\mathrm{m}$. This domain is uniformly
discretized into an $I\times J$ grids with $I=J=100$. The ground-truth
RSS matrix, $\bm{Z}$, is generated according to \eqref{eq:gamma_ij},
incorporating a shadowing component with standard deviation $\sigma\in[1,4]$
dB. The path gain $\beta_{m}$ and element-wise distances $\{d_{m}^{(n)}\}$
are computed using the spherical wavefront model as defined in \eqref{eq:beta m}
and \eqref{eq:r m n}, respectively.

We quantify the performance of the near-field radio map reconstruction
using the \ac{nmse}, defined as:
\begin{equation}
\text{NMSE}=\frac{\|10^{\bm{Z}/10}-10^{\hat{\bm{Z}}/10}\|_{F}^{2}}{\|10^{\bm{Z}/10}\|_{F}^{2}},\label{eq:nmse}
\end{equation}
where $\bm{Z}$ and $\hat{\bm{Z}}$ are the ground-truth and estimated
RSS matrices in decibels (dB), respectively, and $\|\cdot\|_{F}$
denotes the Frobenius norm.

In the proposed method, the tolerance parameter $\delta$ in \eqref{eq:MC}
is automatically determined via the robust Huberized LOOCV scheme,
with the Huber threshold $\varsigma_{H}$ set to the median absolute
deviation of the LOOCV residuals $\varsigma_{H}=\text{median}\Bigl(\lvert e_{i}-\text{median}(e_{j})\rvert\Bigr).$

We compare the proposed approach against four baselines:
\begin{itemize}
\item RBF Interpolation: This method applies multiquadric RBF interpolation
to each angular slice independently to estimate the missing RSS values,
without leveraging global matrix structure.
\item Matrix Completion via \ac{nnm} (MC-NNM)\cite{CanPla:J10}: A standard
matrix completion approach that solves \eqref{eq:MC} using \ac{nnm}
without RBF-based interpolation prior.
\item Local Polynomial Regression (LPR): A first-order local polynomial
model \cite{Fan:b96, VerFunRaj:J16} with a Gaussian weighting kernel.
\item LPR-assisted Matrix Completion (LPR-MC): A two-stage method where
missing RSS values are first estimated using \ac{lpr}, followed by
a refinement step via an \ac{nnm}-based matrix completion algorithm
\cite{CheWanZha:J23, SunChe:J22}.
\end{itemize}

\subsection{Analysis of RBF Kernel Selection}

We first validate the design choices for our RBF interpolation module.
\begin{figure}
\includegraphics{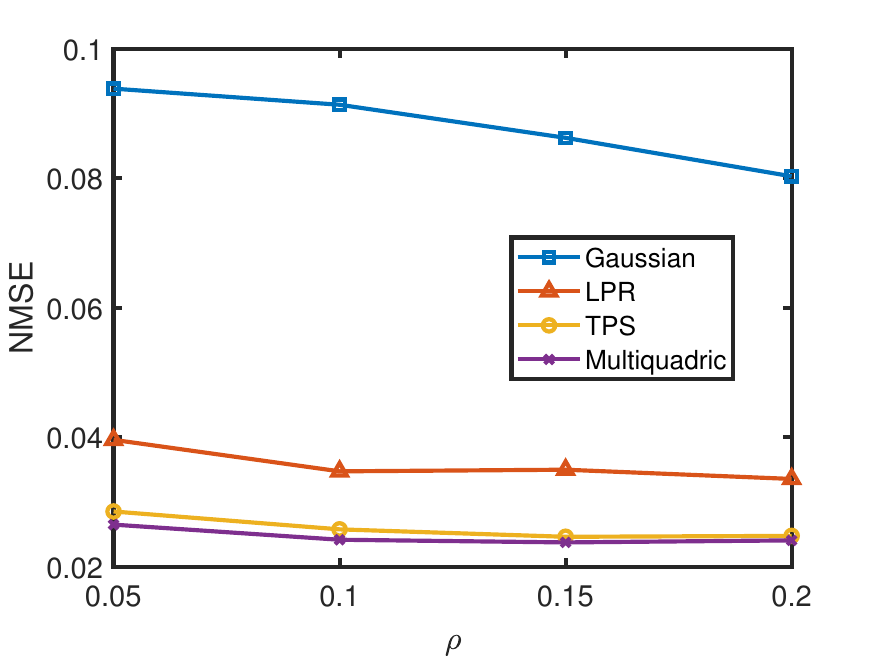}\caption{\label{fig:Reconstruction-NMSE-versus sr under interpolation}Reconstruction
NMSE versus sampling ratio for different interpolation methods. The
Quadratic RBF outperforms TPS, Gaussian RBF and LPR, particularly
at low sampling ratios.}
\end{figure}
\begin{figure}
\includegraphics{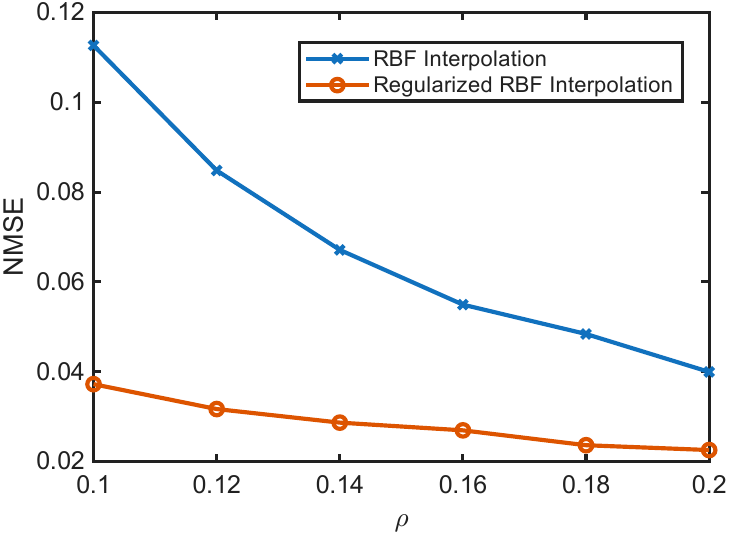}

\caption{\label{fig:regularized RBF}Comparison of reconstruction NMSE versus
sampling ratio $\rho$ for regularized and non-regularized RBF interpolation.
The proposed regularization achieves substantial gains at low sampling
densities by improving stability and extrapolation robustness.}
\end{figure}

We compare four interpolation schemes: multiquadric RBF ($\phi(d)=\sqrt{1+\epsilon d^{2}}$
with $\epsilon=1$), Gaussian RBF ($\phi(d)=\exp(-\epsilon d^{2})$),
\ac{tps} ($\phi(d)=d^{2}\log d$), and \ac{lpr}. Interpolation is
performed on the radial dimension given each angle, using a uniformly
sampled subset of measurements. To ensure a fair comparison, hyperparameters
for all four interpolation schemes are carefully tuned. In particular,
the LPR adopts a first-order local polynomial model with a Gaussian
kernel for weighting.

Fig. \ref{fig:Reconstruction-NMSE-versus sr under interpolation}
plots the reconstruction \ac{nmse} versus the sampling ratio $\rho=0.05-0.2$.
The multiquadric RBF consistently outperforms the other kernels, particularly
at low sampling ratios. Unlike the overly smooth Gaussian and \ac{tps}
kernels, it effectively balances smoothness with the ability to capture
sharp local variations, justifying its use in our framework.

\subsection{Effectiveness of the Regularized RBF Interpolation}

We now examine the performance gain brought by the proposed regularized
RBF interpolation introduced in (\ref{eq:rho_i}). We randomly sample
a subset of entries in each row of the ground-truth RSS matrix $\bm{Z}$,
with the sampling ratio $\rho$ ranging from 0.1 to 0.2.

Fig.~\ref{fig:regularized RBF} compares the reconstruction NMSE
of the standard RBF interpolation and the proposed regularized RBF
interpolation across different sampling ratios. The results clearly
show that the regularized RBF consistently achieves lower NMSE over
all sampling regimes, yielding more than $40$\% NMSE reduction compared
to the unregularized baseline.

This performance improvement is primarily attributed to the inclusion
of a constant term in the RBF formulation. Without this term, the
RBF interpolant must implicitly capture both local variations and
the global offset of the signal. This can lead to instability and
overfitting, especially near the boundaries or in regions with sparse
observations. By decoupling the global offset (handled by the constant
term) from local fluctuations (modeled by the RBFs), the regularized
interpolation becomes more robust and accurate, especially in extrapolation-prone
areas.

\begin{figure}
\includegraphics{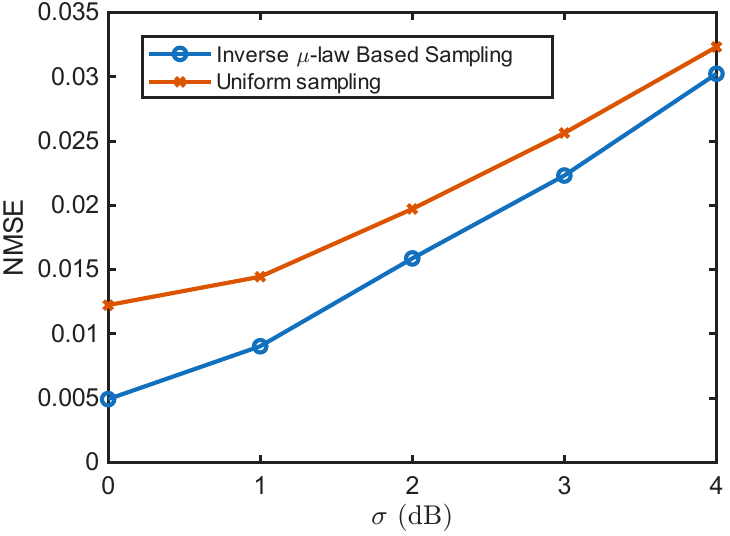}

\caption{\label{fig:mu law}Comparison of reconstruction NMSE versus RSS shadowing
condition $\sigma$ under uniform sampling and inverse $\mu$-law-inspired
non-uniform sampling. The proposed non-uniform sampling strategy consistently
achieves lower NMSE, particularly under low to moderate shadowing.}
\end{figure}

\begin{figure}
\includegraphics{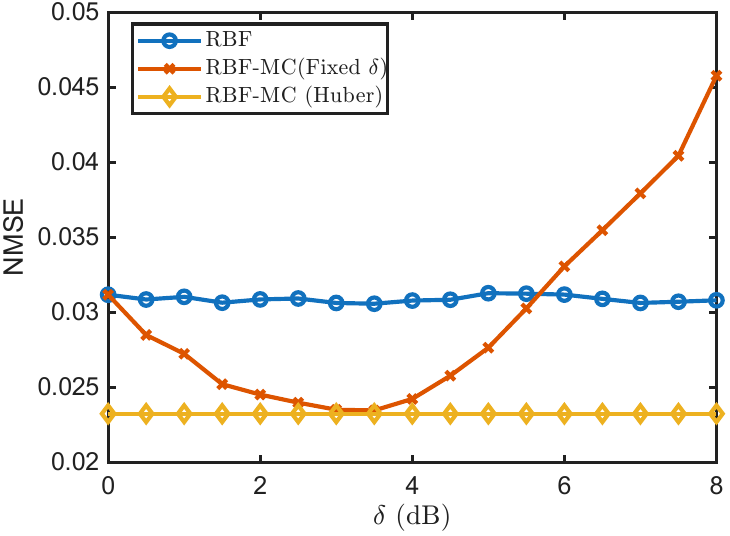}

\caption{\label{fig:huber}Comparison of reconstruction NMSE versus tolerance
parameter $\delta$ for three methods: RBF interpolation (blue), RBF-driven
matrix completion with fixed $\delta$ (red), and the proposed RBF-driven
matrix completion using Huberized LOOCV for automatic $\delta$ estimation
(yellow).}
\end{figure}

\subsection{Performance of the Proposed Inverse $\mu$-Law-Inspired Non-Uniform
Sampling}

In this subsection, we evaluate the effectiveness of the proposed
inverse $\mu$-law-inspired non-uniform sampling strategy for near-field
RSS interpolation. The sampling ratio $\rho=M/(I\times J)$ is fixed
at $0.1$, indicating that only $10\%$ of the RSS measurements over
the angular-distance domain are available for interpolation. To model
the effect of shadowing, we inject $\varepsilon\sim\mathcal{N}(0,\sigma^{2})$
into the RSS values as defined in (\ref{eq:gamma_ij}), and vary the
standard deviation $\sigma$ from $1$ dB to $4$ dB. We adopt a companding
parameter $\mu=15$ for the inverse $\mu$-law transformation to achieve
a dense sampling distribution near the transmitter.

Fig.~\ref{fig:mu law} presents the reconstruction NMSE under both
uniform and inverse $\mu$-law-inspired non-uniform sampling strategies.
Across all shadowing conditions, the proposed method significantly
outperforms the uniform baseline, achieving more than $10\%$ reduction
in NMSE. This improvement is attributed to the fact that the inverse
$\mu$-law transformation allocates more sampling points in regions
of rapid signal variation, such as, near the transmitter, where the
RSS changes more abruptly due to spherical wavefront propagation.
In contrast, uniform sampling may waste measurements in smoother regions,
leading to higher interpolation error in critical areas.

As $\sigma$ increases, the NMSE under both sampling schemes rises
due to the growing impact of the shadowing effect on the measurements.
Nevertheless, the proposed non-uniform strategy maintains its advantage,
demonstrating better robustness to the shadowing effect.

These results confirm that the proposed inverse $\mu$-law-inspired
sampling approach can efficiently allocate limited measurement resources
and substantially enhance interpolation accuracy.

\subsection{Performance of the Huberized LOOCV $\delta$ Estimate}

We examine the impact of the tolerance parameter $\delta$ in the
proposed RBF-driven matrix completion framework and demonstrate the
effectiveness of using the Huberized LOOCV scheme for automatically
$\delta$ estimation. The sampling ratio is fixed at $\rho=0.2$,
and the shadowing standard deviation is set to $\sigma=4$.

To assess sensitivity, we vary the tolerance parameter $\delta$ in
the matrix completion problem (\ref{eq:MC}) from 0 to 8. The resulting
reconstruction NMSE for fixed-$\delta$ configurations is shown as
the red curve in Fig.~\ref{fig:huber}. The performance clearly exhibits
a strong dependence on the choice of $\delta$: excessively small
values result in overfitting to noisy RBF interpolants, while excessively
large values fail to enforce meaningful agreement between the matrix
completion output and the RBF estimates. An optimal range for $\delta$
exists, but identifying it manually is challenging and data-dependent.

To mitigate this issue, we employ the Huberized LOOCV method proposed
in Section~\ref{subsec:Fusing-the-RBF} to automatically determine
a robust $\delta$. The result of this method is shown as the yellow
curve in Fig.~\ref{fig:huber}, which appears flat across all $\delta$
values because $\delta$ is internally estimated rather than externally
set. The performance is consistently superior to both the RBF interpolation
(blue curve) and most fixed-$\delta$ settings, validating the effectiveness
and robustness of the Huberized estimator.

These results confirm that (i) the performance of RBF-driven matrix
completion is highly sensitive to the $\delta$ parameter, and (ii)
the proposed Huberized LOOCV estimator offers a principled and robust
way to automatically select $\delta$, avoiding manual tuning and
enhancing reconstruction accuracy.

\subsection{Performance of the Proposed RBF-Driven Matrix Completion Under Different
Sampling Ratios and Shadowing Conditions}

In this subsection, we evaluate the performance of the proposed RBF-driven
matrix completion method under varying sampling ratios and shadowing
conditions, comparing it with four baseline approaches.
\begin{figure}
\includegraphics{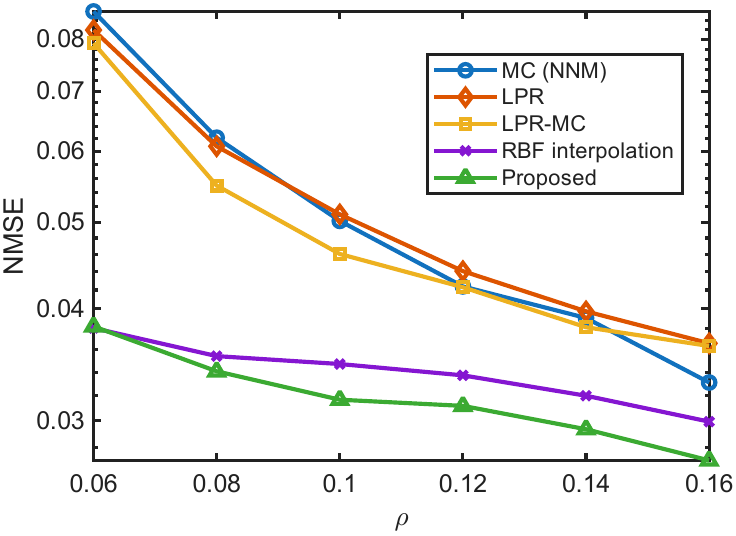}\caption{\label{fig:Reconstruction-NMSE-versus sr}Reconstruction NMSE versus
sampling ratio $\rho$. The proposed RBF-driven matrix completion
method consistently achieves the lowest NMSE across all sampling ratios.
Notably, it offers more than $10$\% improvement compared to baseline
methods when the sampling ratio is moderate.}
\end{figure}
 
\begin{figure}
\includegraphics{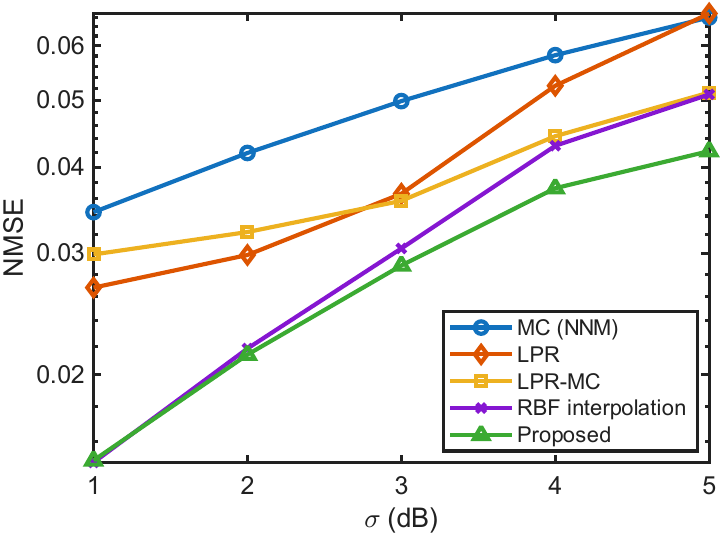}\caption{\label{fig:Reconstruction-NMSE-versus noiselevel}Reconstruction NMSE
versus shadowing standard deviation $\sigma$. The proposed method
consistently achieves the lowest NMSE across varying shadowing levels,
demonstrating strong robustness to the shadowing effect.}
\end{figure}

Fig.~\ref{fig:Reconstruction-NMSE-versus sr} illustrates the reconstruction
NMSE as a function of the sampling ratio $\rho$, ranging from $0.06$
to $0.16$, with a fixed shadowing standard deviation $\sigma=3$.
The proposed method consistently outperforms all baselines, achieving
more than a $10\%$ improvement in NMSE at moderate to high sampling
ratios. However, at very low sampling ratios, the scarcity of measurements
adversely affects the accuracy of the Huberized LOOCV-based parameter
estimation, leading to a degradation in reconstruction performance

Next, we evaluate the robustness of the proposed method against the
shadowing effect. Fig.~\ref{fig:Reconstruction-NMSE-versus sr} illustrates
the reconstruction NMSE as a function of the shadowing standard deviation
$\sigma$, which varies from $1$ to $5$, with the sampling ratio
fixed at $\rho=0.1$. The proposed method consistently outperforms
all baseline approaches across the entire range of shadowing conditions,
achieving more than a $10$\% improvement in NMSE even under severe
shadowing. This robustness is attributed to the Huberized LOOCV-based
selection of the tolerance parameter, which effectively integrates
RBF interpolation with matrix completion. By leveraging both local
smoothness and the global low-rank structure, the proposed method
significantly mitigates the impact of outliers in noisy measurements.

These results confirm the effectiveness of the proposed framework
in leveraging both spatial smoothness and low-rank structure, and
its robustness to practical challenges such as sparse sampling and
the shadowing effect.

To qualitatively assess reconstruction accuracy, we present a visual
comparison of the recovered near-field radio map under a sampling
ratio $\rho=0.1$ and shadowing $\sigma=4$ dB. As shown in Fig.~\ref{fig:Visual-plot-of},
the proposed method exhibits the closest resemblance to the ground
truth, with sharper spatial features and reduced artifacts compared
to the baseline approaches.

\begin{figure}
\includegraphics{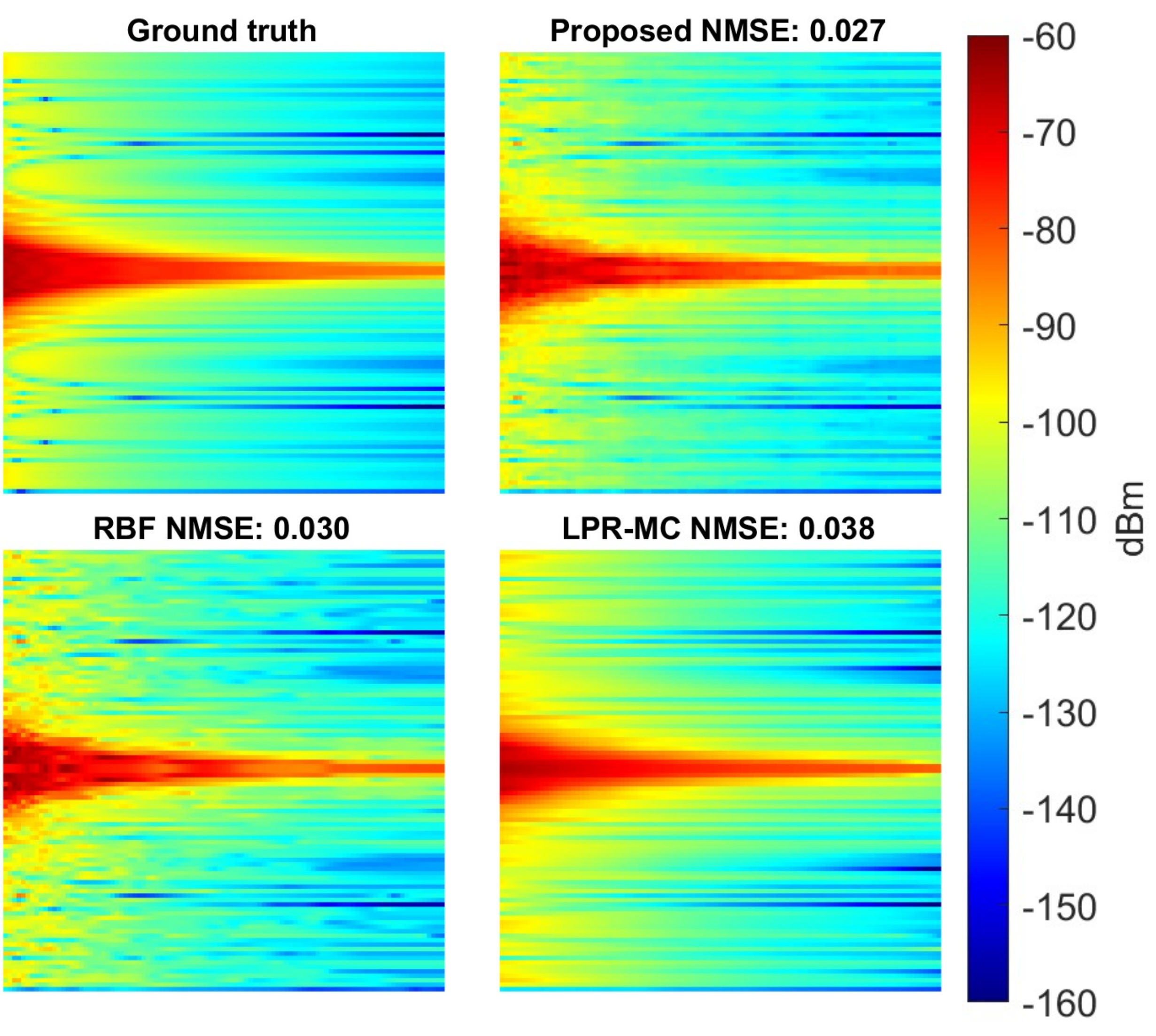}\caption{\label{fig:Visual-plot-of}Visual comparison of reconstruction performance
at sampling ratio $\rho=0.1$ and shadowing $\sigma=4$ dB. The horizontal
axis represents the distance $d$, and the vertical axis represents
the angular direction $\theta$.}
\end{figure}

\section{Conclusion\label{sec:Conclusion}}

In this paper, we proposed a novel structure -aware RBF-assisted matrix
completion framework tailored explicitly for near-field radio map
reconstruction in millimeter-wave XL-MIMO systems. Our approach effectively
integrates local spatial variations through a regularized RBF interpolation
with a global low-rank structure via \ac{nnm}, significantly outperforming
conventional techniques. The introduction of a constant term in the
regularized RBF interpolation notably enhanced accuracy near the boundaries,
overcoming the limitations prevalent in standard interpolation methods.
Leveraging theoretical interpolation error analysis, an adaptive inverse
$\mu$-law-inspired non-uniform sampling strategy is designed to optimally
allocate measurement points, focusing sampling density in areas with
rapid RSS variations near the transmitter. Moreover, a robust Huberized
LOOCV scheme is introduced to adaptively select the tolerance parameter,
effectively bridging local spatial correlation from RBF interpolation
and global low-rank characteristics inherent in matrix completion.
Extensive simulation results validated the superior performance and
robustness of the proposed method, achieving more than $10$\% improvement
in NMSE compared to existing interpolation and matrix completion techniques
under various sampling densities and shadowing conditions.


\appendices{}


\section{Proof of Theorem \ref{thm:The-partial-derivatives}\label{sec:Proof-of-Lemma non smooth}}

Given the near-field array factor
\[
S(d,\vartheta)=\frac{\lambda}{4\pi d}\sum_{n=0}^{N-1}v_{n}e^{-\jmath\frac{2\pi}{\lambda}\,d^{(n)}(d,\vartheta)},
\]
where
\[
d^{(n)}(d,\vartheta)=\sqrt{d^{2}+\delta_{n}^{2}\lambda^{2}/4-d\delta_{n}\lambda\cos\vartheta},
\]
we analyze the sensitivity of $|S(d,\vartheta)|$ \ac{wrt} both the
radial distance $d$ and angle $\vartheta$. At certain critical points
where the phase contributions from all array elements precisely cancel
each other (i.e., $S(d,\vartheta)=0$), the magnitude $|S(d,\vartheta)|$
becomes non-differentiable, despite $S(d,\vartheta)$ itself remaining
smooth. In the following derivation, unless otherwise noted, we consider
the differentiable region only.

We first show that $S(d,\vartheta)$ is smooth in both $d$ and $\vartheta$.
The expression $d^{(n)}$ involves a square root of a quadratic polynomial,
which is smooth because $d^{(n)}(d,\vartheta)$ represents the distance
from $n$th antenna to the receiver. The function
\[
g_{n}(d,\vartheta)=-\jmath\tfrac{2\pi}{\lambda}d^{(n)}(d,\vartheta)
\]
is a smooth composition of the linear factor $-\jmath\tfrac{2\pi}{\lambda}$
and the smooth square root, and hence is smooth itself. Consequently,
the complex exponential $f_{n}(d,\vartheta)=e^{g_{n}(d,\vartheta)}$
is also smooth by the chain rule. Since $\ensuremath{S(d,\vartheta)}$
is a finite sum of smooth functions, $S(d,\vartheta)=\sum_{n=0}^{N-1}f_{n}(d,\vartheta),$
it follows that $S(d,\vartheta)\in C^{\infty}$, i.e., it is infinitely
differentiable in both $d$ and $\vartheta$.

Then, we perform the derivative of $S(d,\vartheta)$ \ac{wrt} both
$d$ and $\vartheta$. The angular derivative of $S(d,\vartheta)$
\ac{wrt} $\vartheta$ is computed as
\begin{align*}
\frac{\partial S}{\partial\vartheta} & =-\frac{\jmath}{2d}\sum_{n=0}^{N-1}v_{n}\frac{\partial d^{(n)}}{\partial\vartheta}e^{-\jmath\frac{2\pi}{\lambda}d^{(n)}}.
\end{align*}
Taking the magnitude, we have
\begin{align*}
\left|\frac{\partial S}{\partial\vartheta}\right| & \leq\left|-\frac{\jmath}{2d}\right|\sum_{n=0}^{N-1}|v_{n}|\left|\frac{\partial d^{(n)}}{\partial\vartheta}\right||e^{-\jmath\frac{2\pi}{\lambda}d^{(n)}}|\\
 & \leq\frac{\text{max}_{n}|v_{n}|}{2d}\sum_{n=0}^{N-1}\left|\frac{\partial d^{(n)}}{\partial\vartheta}\right|\\
 & \leq\frac{1}{2d}\sum_{n=0}^{N-1}\left|\frac{\partial d^{(n)}}{\partial\vartheta}\right|.
\end{align*}
Since $\partial d^{(n)}/\partial\vartheta=d\lambda\delta_{n}/2/d^{(n)}\sin\vartheta$
and $\delta_{n}$ typically grows linearly with $n$, we have
\[
\left|\frac{\partial S}{\partial\vartheta}\right|\leq\frac{\lambda|\sin\vartheta|}{4}\sum_{n=0}^{N-1}|\delta_{n}|/d^{(n)}.
\]
Define $d_{\text{min}}=\text{min}_{n}d^{(n)}$. Since $\sum_{n=0}^{N-1}|\delta_{n}|\leq N^{2}/4+1$,
then
\[
\left|\frac{\partial S}{\partial\vartheta}\right|\leq\frac{\lambda|\sin\vartheta|}{16d_{\text{min}}}(N^{2}+4)\leq\frac{5\lambda|\sin\vartheta|}{16d_{\text{min}}}N^{2},
\]
since $N>1$.

The derivative over distance $d$ is given by
\begin{align}
\frac{\partial S}{\partial d} & =-\frac{\lambda}{4\pi d^{2}}\sum_{n=0}^{N-1}v_{n}e^{-\jmath\frac{2\pi}{\lambda}d^{(n)}}\label{eq:s/d}\\
 & \qquad-\frac{1}{2d}\jmath\!\sum_{n=0}^{N-1}v_{n}\frac{\partial d^{(n)}}{\partial d}\,e^{-\jmath\frac{2\pi}{\lambda}d^{(n)}},\nonumber 
\end{align}
with $\partial d^{(n)}/\partial d=(2d-\delta_{n}\lambda\cos\vartheta)/(2d^{(n)}).$

Taking the magnitude, we have
\begin{align}
\left|\frac{\partial S}{\partial d}\right| & \leq\left|\frac{\lambda}{4\pi d^{2}}\sum_{n=0}^{N-1}v_{n}e^{-\jmath\frac{2\pi}{\lambda}d^{(n)}}\right|\label{eq:s/d-1}\\
 & \qquad+\left|\frac{1}{2d}\jmath\!\sum_{n=0}^{N-1}v_{n}\frac{\partial d^{(n)}}{\partial d}\,e^{-\jmath\frac{2\pi}{\lambda}d^{(n)}}\right|,\nonumber \\
 & <\frac{\lambda N}{4\pi d^{2}}+\left|\frac{1}{2d}\sum_{n=0}^{N-1}v_{n}\frac{2d-\delta_{n}\lambda\cos\vartheta}{2d^{(n)}}\,e^{-\jmath\frac{2\pi}{\lambda}d^{(n)}}\right|
\end{align}
Nest, we focus on bounding the second term $\left|\frac{1}{2d}\sum_{n=0}^{N-1}v_{n}\frac{2d-\delta_{n}\lambda\cos\vartheta}{2d^{(n)}}\,e^{-\jmath\frac{2\pi}{\lambda}d^{(n)}}\right|$.

Applying the triangle inequality and using the bound $|v_{n}|\leq$1,
we have:
\begin{align}
 & \left|\frac{1}{2d}\sum_{n=0}^{N-1}v_{n}\frac{2d-\delta_{n}\lambda\cos\vartheta}{2d^{(n)}}\,e^{-\jmath\frac{2\pi}{\lambda}d^{(n)}}\right|\nonumber \\
\leq & \frac{1}{2d}\sum_{n=0}^{N-1}\left|v_{n}\right|\left|\frac{2d-\delta_{n}\lambda\cos\vartheta}{2d^{(n)}}\right|\nonumber \\
\leq & \frac{1}{2d}\sum_{n=0}^{N-1}\left|\frac{2d-\delta_{n}\lambda\cos\vartheta}{2d^{(n)}}\right|.\label{eq:bound}
\end{align}
To further simplify \eqref{eq:bound}, we first analyze the squared
distance $(d^{(n)})^{2}$ by completing the square.
\begin{align}
 & (d^{(n)})^{2}\nonumber \\
 & =d^{2}+\frac{\delta_{n}^{2}\lambda^{2}}{4}-d\delta_{n}\lambda\cos\vartheta\nonumber \\
 & =\left(d^{2}-d\delta_{n}\lambda\cos\vartheta+\frac{\delta_{n}^{2}\lambda^{2}}{4}\cos^{2}\vartheta\right)+\frac{\delta_{n}^{2}\lambda^{2}}{4}-\frac{\delta_{n}^{2}\lambda^{2}}{4}\cos^{2}\vartheta\nonumber \\
 & =\left(d-\frac{\delta_{n}\lambda}{2}\cos\vartheta\right)^{2}+\frac{\delta_{n}^{2}\lambda^{2}}{4}(1-\cos^{2}\vartheta)\nonumber \\
 & =\left(d-\frac{\delta_{n}\lambda}{2}\cos\vartheta\right)^{2}+\left(\frac{\delta_{n}\lambda}{2}\sin\vartheta\right)^{2}.\label{eq:dn}
\end{align}
Now, we substitute \eqref{eq:dn} into \eqref{eq:bound}:
\begin{align*}
\left|\frac{2d-\delta_{n}\lambda\cos\vartheta}{2d^{(n)}}\right| & =\left|\frac{2\left(d-\frac{\delta_{n}\lambda}{2}\cos\vartheta\right)}{2\sqrt{\left(d-\frac{\delta_{n}\lambda}{2}\cos\vartheta\right)^{2}+\left(\frac{\delta_{n}\lambda}{2}\sin\vartheta\right)^{2}}}\right|\\
 & =\frac{\left|d-\frac{\delta_{n}\lambda}{2}\cos\vartheta\right|}{\sqrt{\left(d-\frac{\delta_{n}\lambda}{2}\cos\vartheta\right)^{2}+\left(\frac{\delta_{n}\lambda}{2}\sin\vartheta\right)^{2}}}.
\end{align*}

To simplify, let $X=d-\frac{\delta_{n}\lambda}{2}\cos\vartheta$ and
$Y=\frac{\delta_{n}\lambda}{2}\sin\vartheta$. The expression becomes
$\frac{|X|}{\sqrt{X^{2}+Y^{2}}}$.

Since $Y^{2}\geq0$, the denominator $\sqrt{X^{2}+Y^{2}}$ is always
greater than or equal to $\sqrt{X^{2}}=|X|$. This leads to the following
strict inequality:
\begin{equation}
\frac{|X|}{\sqrt{X^{2}+Y^{2}}}\leq1.
\end{equation}

This demonstrates that $\left|\frac{2d-\delta_{n}\lambda\cos\vartheta}{2d^{(n)}}\right|\leq1$
for all $n$, $d$, and $\vartheta$.

Finally, we have
\begin{equation}
\frac{1}{2d}\sum_{n=0}^{N-1}\left|\frac{2d-\delta_{n}\lambda\cos\vartheta}{2d^{(n)}}\right|\leq\frac{1}{2d}\sum_{n=0}^{N-1}(1)=\frac{N}{2d}.
\end{equation}

Consequently, we conclude that the range derivative satisfies: 
\begin{equation}
\left|\frac{\partial S}{\partial d}\right|\leq\left(\frac{\lambda}{4\pi d^{2}}+\frac{1}{2d}\right)N\label{eq:d/dn}
\end{equation}

The derivative of the magnitude satisfies
\begin{equation}
\frac{\partial|S|}{\partial x}=\Bigl|\frac{\partial S}{\partial x}\Bigr|\,\cos\phi_{x},\qquad x\in\{d,\vartheta\},
\end{equation}
where $\cos\phi_{x}$ is a phase-alignment factor with $|\cos\phi_{x}|\le1$
and is independent of the array size $N$. Thus, the scalings in $N$
of $\bigl|\partial|S|/\partial x\bigr|$ and $\bigl|\partial S/\partial x\bigr|$
are identical. Thus, we have 
\begin{align}
\frac{\partial|S|}{\partial\vartheta} & \leq\text{cos}(\phi_{\vartheta})\frac{5\lambda|\sin\vartheta|}{16d_{\text{min}}}N^{2}\label{eq: sensitivity angle-1}\\
\frac{\partial|S|}{\partial d} & \leq\text{cos}(\phi_{d})\left(\frac{\lambda}{4\pi d^{2}}+\frac{1}{2d}\right)N.\label{eq:sensitivity distance-1}
\end{align}

\section{\label{sec:Proof-of-Proposition}Proof of Theorem \ref{thm:Suppose-that-}}

Consider the function $f(d)=(1+d)^{\beta}$. Its $\ell$th derivative
is given by
\[
f^{(\ell)}(d)=\beta(\beta-1)\cdots(\beta-\ell+1)(1+d)^{\beta-\ell},
\]
which can be rewritten as
\begin{equation}
f^{(\ell)}(d)=\ell!\prod_{j=0}^{\ell-1}\frac{\beta-j}{j+1}\cdot(1+d)^{\beta-\ell}.\label{eq:flr}
\end{equation}
We now bound each term in the product in (\ref{eq:flr}). For each
index $0\le j\le\ell-1$, we have
\[
\left|\frac{\beta-j}{j+1}\right|=\left|1-\frac{\beta+1}{j+1}\right|\le1+\frac{|\beta+1|}{j+1}\le1+|\beta+1|.
\]
Denote $M_{0}=1+|\beta+1|$. Since there are $\ell$ factors in the
product, the over all bound becomes
\[
\left|\ell!\prod_{j=0}^{\ell-1}\frac{\beta-j}{j+1}\right|\leq\ell!M_{0}^{\ell}.
\]
From (\ref{eq:flr}), we thus have
\begin{equation}
\left|f^{(\ell)}(d)\right|\le\ell!\,M_{0}^{\ell}|(1+d)^{\beta-\ell}|.\label{eq:bound1}
\end{equation}
Since $1+d$ is bounded, as $d\in\mathcal{D}$, it is clearly that
$|(1+d)^{\beta-\ell}|\leq M_{1}$, thus, combining this with (\ref{eq:bound1})
gives
\[
\left|f^{(\ell)}(d)\right|\le\ell!M_{1}M_{0}^{\ell}\leq\ell!C_{1}^{\ell},
\]
where $C_{1}=M_{1}^{1/\ell}M_{0}$.

This derivative growth condition satisfies the hypothesis of Theorem
11.22 in \cite{Wen:B04}, which guarantees that, for sufficiently
small fill distance $h$, the interpolation error decays exponentially
\[
\left\Vert f-\rho\right\Vert _{L_{\infty}(\mathcal{D})}\leq e^{-c/h}|f|_{\mathcal{N}_{\phi}(\mathcal{D})}.
\]
This completes the proof of Theorem \ref{thm:Suppose-that-}.

\bibliographystyle{IEEEtran}
\bibliography{IEEEabrv,StringDefinitions,JCgroup,ChenBibCV}

\end{document}